\documentclass[english,DIV12,11pt]{article}
\usepackage[T1]{fontenc}
\usepackage[latin9]{inputenc}
\usepackage{color}
\usepackage{float}
\usepackage{fancybox}
\usepackage{calc}
\usepackage{mathrsfs}
\usepackage{amsmath}
\usepackage{amssymb}
\usepackage{graphicx}

\makeatletter

\floatstyle{ruled}
\newfloat{algorithm}{tbp}{loa}
\providecommand{\algorithmname}{Algorithm}
\floatname{algorithm}{\protect\algorithmname}

\usepackage[noblocks]{authblk}

\usepackage{amsmath,amsthm,amssymb, cancel, esint, mathdots}

\usepackage{thmtools}

\usepackage{shadethm}
\newshadetheorem{shnote}{\large \bfseries Note}[section]

\newtheorem{theorem}{Theorem}
\newtheorem{lemma}[theorem]{Lemma}

\newtheorem{definition}[theorem]{Definition}

\newtheorem{notation}[theorem]{Notation}

\theoremstyle{definition}

\usepackage{graphicx}
\usepackage[table, dvipsnames,svgnames,x11names]{xcolor} 	
\usepackage{transparent} 
\colorlet{shadecolor}{gray!20}   
\colorlet{shadedcolor}{gray!50}
\colorlet{shadycolor}{gray!5}
\colorlet{shadestcolor}{gray!75}
\definecolor{webgreen}{rgb}{0, 0.5, 0} 
\definecolor{webblue}{rgb}{0, 0, 0.5}
\definecolor{webred}{rgb}{0.5, 0, 0}
\definecolor{Periwinkle}{RGB}{120,80,255}
\definecolor{LightViolet}{RGB}{190,180,250}
\definecolor{LighterOrange}{rgb}{0.5,0.25,0}
\colorlet{headbg}{CornflowerBlue!50}
\colorlet{headtitle}{SlateBlue!75!black}
\definecolor{wwqqzz}{rgb}{0.4,0.0,0.6} 
\definecolor{yqqqyq}{rgb}{0.5019607843137255,0.0,0.5019607843137255}

\definecolor{hexcolor0xffb6c1}{rgb}{1.000,0.714,0.757}
\definecolor{hexcolor0x000000}{rgb}{0.000,0.000,0.000}
\definecolor{hexcolor0x48d1cc}{rgb}{0.282,0.820,0.800}
\definecolor{hexcolor0x000000}{rgb}{0.000,0.000,0.000}
\definecolor{hexcolor0xdda0dd}{rgb}{0.867,0.627,0.867}
\definecolor{hexcolor0x000000}{rgb}{0.000,0.000,0.000}
\definecolor{hexcolor0xa588e1}{rgb}{0.647,0.533,0.882}
\definecolor{hexcolor0x000000}{rgb}{0.000,0.000,0.000}
\definecolor{hexcolor0x87abec}{rgb}{0.529,0.671,0.925}
\definecolor{hexcolor0x000000}{rgb}{0.000,0.000,0.000}
\definecolor{hexcolor0xdc7fe5}{rgb}{0.863,0.498,0.898}
\definecolor{hexcolor0x000000}{rgb}{0.000,0.000,0.000}
\definecolor{hexcolor0x52e594}{rgb}{0.322,0.898,0.580}
\definecolor{hexcolor0x000000}{rgb}{0.000,0.000,0.000}
\definecolor{hexcolor0xff00ff}{rgb}{1.000,0.000,1.000}
\definecolor{hexcolor0x000000}{rgb}{0.000,0.000,0.000}
\definecolor{hexcolor0xff5bad}{rgb}{1.000,0.357,0.678}
\definecolor{hexcolor0x000000}{rgb}{0.000,0.000,0.000}
\definecolor{hexcolor0xc9a0dc}{rgb}{0.788,0.627,0.863}
\definecolor{hexcolor0x000000}{rgb}{0.000,0.000,0.000}
\definecolor{hexcolor0xe396e0}{rgb}{0.890,0.588,0.878}
\definecolor{hexcolor0x000000}{rgb}{0.000,0.000,0.000}
\definecolor{hexcolor0xe693ff}{rgb}{0.902,0.576,1.000}
\definecolor{hexcolor0x000000}{rgb}{0.000,0.000,0.000}
\definecolor{hexcolor0xe093c7}{rgb}{0.878,0.576,0.780}
\definecolor{hexcolor0x000000}{rgb}{0.000,0.000,0.000}
\definecolor{hexcolor0xb59bb1}{rgb}{0.710,0.608,0.694}
\definecolor{hexcolor0x000000}{rgb}{0.000,0.000,0.000}
\definecolor{hexcolor0xffffff}{rgb}{1.000,1.000,1.000}
\definecolor{hexcolor0x000000}{rgb}{0.000,0.000,0.000}
\definecolor{hexcolor0x466de1}{rgb}{0.275,0.427,0.882}
\definecolor{hexcolor0x000000}{rgb}{0.000,0.000,0.000}
\definecolor{hexcolor0x3d94ff}{rgb}{0.239,0.580,1.000}
\definecolor{hexcolor0x000000}{rgb}{0.000,0.000,0.000}
\definecolor{hexcolor0x4f91d3}{rgb}{0.310,0.569,0.827}
\definecolor{mylapuslazuli}{RGB}{59,99,139}
\definecolor{hexcolor0x000000}{rgb}{0.000,0.000,0.000}
\definecolor{hexcolor0x1e96ff}{rgb}{0.118,0.588,1.000}
\definecolor{hexcolor0x000000}{rgb}{0.000,0.000,0.000}
\definecolor{hexcolor0x93ccea}{rgb}{0.576,0.800,0.918}
\definecolor{hexcolor0x000000}{rgb}{0.000,0.000,0.000}
\definecolor{hexcolor0x00bfff}{rgb}{0.000,0.749,1.000}
\definecolor{hexcolor0x000000}{rgb}{0.000,0.000,0.000}
\definecolor{hexcolor0xb2ffff}{rgb}{0.698,1.000,1.000}
\definecolor{hexcolor0x000000}{rgb}{0.000,0.000,0.000}
\definecolor{hexcolor0x51f7ff}{rgb}{0.318,0.969,1.000}
\definecolor{hexcolor0x000000}{rgb}{0.000,0.000,0.000}
\definecolor{hexcolor0x004bb0}{rgb}{0.000,0.294,0.690}
\definecolor{hexcolor0x000000}{rgb}{0.000,0.000,0.000}
\definecolor{hexcolor0x1a33ac}{rgb}{0.102,0.200,0.675}
\definecolor{hexcolor0x000000}{rgb}{0.000,0.000,0.000}
\definecolor{hexcolor0x6f00ff}{rgb}{0.435,0.000,1.000}
\definecolor{hexcolor0x000000}{rgb}{0.000,0.000,0.000}
\definecolor{hexcolor0x00cece}{rgb}{0.000,0.808,0.808}
\definecolor{hexcolor0x000000}{rgb}{0.000,0.000,0.000}
\definecolor{hexcolor0x7fffd4}{rgb}{0.498,1.000,0.831}
\definecolor{hexcolor0x000000}{rgb}{0.000,0.000,0.000}
\definecolor{hexcolor0x00fa9a}{rgb}{0.000,0.980,0.604}
\definecolor{hexcolor0x000000}{rgb}{0.000,0.000,0.000}
\definecolor{hexcolor0x00ffff}{rgb}{0.000,1.000,1.000}
\definecolor{hexcolor0x000000}{rgb}{0.000,0.000,0.000}
\definecolor{hexcolor0x00b7eb}{rgb}{0.000,0.718,0.922}
\definecolor{hexcolor0x000000}{rgb}{0.000,0.000,0.000}
\definecolor{hexcolor0x00a693}{rgb}{0.000,0.651,0.576}
\definecolor{hexcolor0x000000}{rgb}{0.000,0.000,0.000}
\definecolor{hexcolor0xaaf0d1}{rgb}{0.667,0.941,0.820}
\definecolor{hexcolor0x000000}{rgb}{0.000,0.000,0.000}
\definecolor{hexcolor0x2dbb94}{rgb}{0.176,0.733,0.580}
\definecolor{hexcolor0x000000}{rgb}{0.000,0.000,0.000}
\definecolor{hexcolor0xc6ff1a}{rgb}{0.776,1.000,0.102}
\definecolor{hexcolor0x79f200}{rgb}{0.475,0.949,0.000}
\definecolor{hexcolor0x000000}{rgb}{0.000,0.000,0.000}
\definecolor{hexcolor0xdfff00}{rgb}{0.875,1.000,0.000}
\definecolor{hexcolor0x000000}{rgb}{0.000,0.000,0.000}
\definecolor{hexcolor0x50c878}{rgb}{0.314,0.784,0.471}
\definecolor{hexcolor0x000000}{rgb}{0.000,0.000,0.000}
\definecolor{hexcolor0x008080}{RGB}{0,100,100} 
\definecolor{hexcolor0x000000}{rgb}{0.000,0.000,0.000}
\definecolor{hexcolor0x0f527e}{rgb}{0.059,0.322,0.494}
\definecolor{hexcolor0x000000}{rgb}{0.000,0.000,0.000}
\definecolor{hexcolor0x1034a6}{rgb}{0.063,0.204,0.651}
\definecolor{hexcolor0x000000}{rgb}{0.000,0.000,0.000}
\definecolor{hexcolor0x003153}{rgb}{0.000,0.192,0.325}
\definecolor{hexcolor0x000000}{rgb}{0.000,0.000,0.000}
\definecolor{hexcolor0x0014a8}{rgb}{0.000,0.078,0.659}
\definecolor{hexcolor0x000000}{rgb}{0.000,0.000,0.000}
\definecolor{hexcolor0x3f00ff}{rgb}{0.247,0.000,1.000}
\definecolor{hexcolor0x1560bd}{rgb}{0.082,0.376,0.741}
\definecolor{hexcolor0x000000}{rgb}{0.000,0.000,0.000}
\definecolor{hexcolor0xce240a}{rgb}{0.808,0.141,0.039}
\definecolor{hexcolor0x000000}{rgb}{0.000,0.000,0.000}
\definecolor{hexcolor0xfa8072}{rgb}{0.980,0.502,0.447}
\definecolor{hexcolor0x000000}{rgb}{0.000,0.000,0.000}
\definecolor{hexcolor0xf19cbb}{rgb}{0.945,0.612,0.733}
\definecolor{hexcolor0x000000}{rgb}{0.000,0.000,0.000}
\definecolor{hexcolor0xff355e}{rgb}{1.000,0.208,0.369}
\definecolor{hexcolor0xab274f}{rgb}{0.671,0.153,0.310}
\definecolor{hexcolor0x000000}{rgb}{0.000,0.000,0.000}
\definecolor{hexcolor0xca1f7b}{rgb}{0.792,0.122,0.482}
\definecolor{hexcolor0x000000}{rgb}{0.000,0.000,0.000}
\definecolor{hexcolor0xffa6c9}{rgb}{1.000,0.651,0.788}
\definecolor{hexcolor0x000000}{rgb}{0.000,0.000,0.000}
\definecolor{hexcolor0xfb607f}{rgb}{0.984,0.376,0.498}
\definecolor{hexcolor0x000000}{rgb}{0.000,0.000,0.000}
\definecolor{hexcolor0xfe28a2}{rgb}{0.996,0.157,0.635}
\definecolor{hexcolor0x000000}{rgb}{0.000,0.000,0.000}
\definecolor{hexcolor0xaa98a9}{rgb}{0.667,0.596,0.663}
\definecolor{hexcolor0x000000}{rgb}{0.000,0.000,0.000}
\definecolor{hexcolor0xbc8f8f}{rgb}{0.737,0.561,0.561}
\definecolor{hexcolor0x000000}{rgb}{0.000,0.000,0.000}
\definecolor{hexcolor0xb57281}{rgb}{0.710,0.447,0.506}
\definecolor{hexcolor0x000000}{rgb}{0.000,0.000,0.000}
\definecolor{hexcolor0x2a52be}{rgb}{0.165,0.322,0.745}
\definecolor{hexcolor0xffcc33}{rgb}{1.000,0.800,0.200}
\definecolor{hexcolor0xf8de7e}{rgb}{0.973,0.871,0.494}
\definecolor{hexcolor0x000000}{rgb}{0.000,0.000,0.000}
\definecolor{hexcolor0xffdead}{rgb}{1.000,0.871,0.678}
\definecolor{hexcolor0x000000}{rgb}{0.000,0.000,0.000}
\definecolor{hexcolor0xffe5b4}{rgb}{1.000,0.898,0.706}
\definecolor{hexcolor0x000000}{rgb}{0.000,0.000,0.000}
\definecolor{hexcolor0xf7e7ce}{rgb}{0.969,0.906,0.808}
\definecolor{hexcolor0x000000}{rgb}{0.000,0.000,0.000}
\definecolor{hexcolor0xff7f00}{rgb}{1.000,0.498,0.000}
\definecolor{hexcolor0x000000}{rgb}{0.000,0.000,0.000}
\definecolor{hexcolor0xffefd5}{rgb}{1.000,0.937,0.835}
\definecolor{hexcolor0xfdbcb4}{rgb}{0.992,0.737,0.706}
\definecolor{hexcolor0xff9966}{rgb}{1.000,0.600,0.400}
\definecolor{hexcolor0xf88379}{rgb}{0.973,0.514,0.475}
\definecolor{hexcolor0x000000}{rgb}{0.000,0.000,0.000}
\definecolor{hexcolor0xff6347}{rgb}{1.000,0.388,0.278}
\definecolor{hexcolor0x000000}{rgb}{0.000,0.000,0.000}
\definecolor{hexcolor0xfbceb1}{rgb}{0.984,0.808,0.694}
\definecolor{hexcolor0x000000}{rgb}{0.000,0.000,0.000}
\definecolor{hexcolor0xd1bea8}{rgb}{0.820,0.745,0.659}
\definecolor{hexcolor0x000000}{rgb}{0.000,0.000,0.000}
\definecolor{hexcolor0xd99058}{rgb}{0.851,0.565,0.345}
\definecolor{hexcolor0x000000}{rgb}{0.000,0.000,0.000}
\definecolor{hexcolor0xe9967a}{rgb}{0.914,0.588,0.478}
\definecolor{hexcolor0x000000}{rgb}{0.000,0.000,0.000}
\definecolor{hexcolor0x9f8b70}{rgb}{0.624,0.545,0.439}
\definecolor{hexcolor0x000000}{rgb}{0.000,0.000,0.000}
\definecolor{hexcolor0xfdee00}{rgb}{0.992,0.933,0.000}
\definecolor{hexcolor0x000000}{rgb}{0.000,0.000,0.000}
\definecolor{hexcolor0xf4a460}{rgb}{0.957,0.643,0.376}
\definecolor{hexcolor0x000000}{rgb}{0.000,0.000,0.000}
\definecolor{hexcolor0xa67b5b}{rgb}{0.651,0.482,0.357}
\definecolor{hexcolor0x000000}{rgb}{0.000,0.000,0.000}
\definecolor{hexcolor0x7b3f00}{rgb}{0.482,0.247,0.000}
\definecolor{hexcolor0xc19a6b}{rgb}{0.757,0.604,0.420}
\definecolor{hexcolor0xe1a95f}{rgb}{0.882,0.663,0.373}
\definecolor{hexcolor0x000000}{rgb}{0.000,0.000,0.000}
\definecolor{hexcolor0xedc9af}{rgb}{0.929,0.788,0.686}
\definecolor{hexcolor0x000000}{rgb}{0.000,0.000,0.000}
\definecolor{hexcolor0xe97451}{rgb}{0.914,0.455,0.318}
\definecolor{hexcolor0x000000}{rgb}{0.000,0.000,0.000}
\definecolor{hexcolor0xa45a52}{rgb}{0.643,0.353,0.322}
\definecolor{hexcolor0x000000}{rgb}{0.000,0.000,0.000}
\definecolor{hexcolor0x990000}{rgb}{0.600,0.000,0.000}
\definecolor{hexcolor0x000000}{rgb}{0.000,0.000,0.000}
\definecolor{hexcolor0xd2b48c}{rgb}{0.824,0.706,0.549}
\definecolor{hexcolor0x000000}{rgb}{0.000,0.000,0.000}
\definecolor{hexcolor0xc19a6b}{rgb}{0.757,0.604,0.420}
\definecolor{hexcolor0xedc9af}{rgb}{0.929,0.788,0.686}
\definecolor{hexcolor0x000000}{rgb}{0.000,0.000,0.000}
\definecolor{hexcolor0xf5deb3}{rgb}{0.961,0.871,0.702}
\definecolor{hexcolor0x000000}{rgb}{0.000,0.000,0.000}
\definecolor{hexcolor0xf5f5dc}{rgb}{0.961,0.961,0.863}
\definecolor{hexcolor0x000000}{rgb}{0.000,0.000,0.000}
\definecolor{hexcolor0xffffcc}{rgb}{1.000,1.000,0.800}
\definecolor{hexcolor0x000000}{rgb}{0.000,0.000,0.000}
\definecolor{hexcolor0xff7f50}{rgb}{1.000,0.498,0.314}
\definecolor{hexcolor0x000000}{rgb}{0.000,0.000,0.000}
\definecolor{hexcolor0xfbec5d}{rgb}{0.984,0.925,0.365}
\definecolor{hexcolor0xfff8e7}{rgb}{1.000,0.973,0.906}
\definecolor{hexcolor0x71a6d2}{rgb}{0.443,0.651,0.824}
\definecolor{hexcolor0xc0c0c0}{rgb}{0.753,0.753,0.753}
\definecolor{hexcolor0x000000}{rgb}{0.000,0.000,0.000}
\definecolor{hexcolor0xc4aead}{rgb}{0.769,0.682,0.678}
\definecolor{hexcolor0xe5e4e2}{rgb}{0.898,0.894,0.886}
\definecolor{hexcolor0xb2beb5}{rgb}{0.698,0.745,0.710}
\definecolor{hexcolor0x8c92ac}{rgb}{0.549,0.573,0.675}
\definecolor{hexcolor0x000000}{rgb}{0.000,0.000,0.000}
\definecolor{hexcolor0x4682b4}{rgb}{0.275,0.510,0.706}
\definecolor{hexcolor0x000000}{rgb}{0.000,0.000,0.000}
\definecolor{hexcolor0x848282}{rgb}{0.518,0.510,0.510}
\definecolor{hexcolor0x000000}{rgb}{0.000,0.000,0.000}
\definecolor{hexcolor0x4B0082}{RGB}{75, 0, 130} 
\definecolor{hexcolor0x66023C}{RGB}{102, 2, 60} 
\definecolor{hexcolor0xE0B0FF}{RGB}{224, 176, 255}
\definecolor{hexcolor0x5A004A}{RGB}{90,0,74}
\definecolor{hexcolor0x9389b3}{RGB}{147, 137, 179}
\definecolor{hexcolor0x70639b}{RGB}{112, 99, 155}
\definecolor{hexcolor0xbeb9d1}{RGB}{190, 185, 209}
\definecolor{hexcolor0x36454F}{RGB}{54, 69, 79}
\definecolor{hexcolor0x343434}{RGB}{52, 52, 52}
\definecolor{myrobinegg}{RGB}{42,160,160}
\definecolor{lightmayanblue}{RGB}{54,205,255}
\definecolor{myviolet}{RGB}{129,88,211}
\definecolor{mypurple}{RGB}{102,2,60}

\usepackage{hyperref}

\usepackage{pgf}
\usepackage{tikz}
\usetikzlibrary{decorations.pathreplacing,calc,arrows,patterns}

\usepackage{mathrsfs}
\usepackage{amsmath, amssymb, cancel, esint, mathdots}
\usepackage{mathtools}      
\usepackage{bm} 

\DeclarePairedDelimiter\abs{\lvert}{\rvert}%
\DeclarePairedDelimiter\norm{\lVert}{\rVert}%
\makeatletter
\let\oldabs\abs
\def\abs{\@ifstar{\oldabs}{\oldabs*}}
\let\oldnorm\norm
\def\norm{\@ifstar{\oldnorm}{\oldnorm*}}
\makeatother
\usepackage{mhchem,stackrel, stmaryrd}

\usepackage{microtype}
\usepackage{fixmath} 
\usepackage{nicefrac} 
\usepackage{mleftright}
\usepackage{fixltx2e}  
 \usepackage{capt-of}
\usepackage{enumitem}
\setlist{nolistsep}
\usepackage{colonequals}
\usepackage{amscd,pifont,latexsym }
\usepackage{soul}
\usepackage{relsize}			

\usepackage{todonotes}

\usepackage{algorithmicx}
\usepackage{algorithm}
\usepackage{listings}
\usepackage{algcompatible}
\usepackage{algpseudocode}

\algrenewcommand{\algorithmiccomment}[1]{$\varcopyright\triangleright$ #1}

\makeatletter
\algnewcommand{\StructComment}[1]{$\triangleright$ \{#1\}}
\makeatother

\makeatletter
\algnewcommand{\LineComment}[1]{\State \(\varcopyright\triangleright\) #1}
\makeatother

\algblock[Name]{StructBegin}{StructEnd}

\algblockdefx[NAME]{START}{END}%
[2][Unknown]{\textbf{StructStart} #1[#2]}%
{\textbf{StructEnd}}

\makeatletter
\renewcommand{\ALG@beginalgorithmic}{\sffamily}
\makeatother

\makeatletter
\def\therule{\makebox[\algorithmicindent][l]{\hspace*{.5em}\vrule height .75\baselineskip depth .25\baselineskip}}%

\newtoks\therules
\therules={}
\def\appendto#1#2{\expandafter#1\expandafter{\the#1#2}}
\def\gobblefirst#1{
  #1\expandafter\expandafter\expandafter{\expandafter\@gobble\the#1}}%
\def\LState{\State\unskip\the\therules}

\def\pushindent{\appendto\therules\therule}%
\def\popindent{\gobblefirst\therules}%
\def\printindent{\unskip\the\therules}%
\def\printandpush{\printindent\pushindent}%
\def\popandprint{\popindent\printindent}%

\algdef{SE}[WHILE]{While}{EndWhile}[1]
  {\printandpush\algorithmicwhile\ #1\ \algorithmicdo}
  {\popandprint\algorithmicend\ \algorithmicwhile}%
\algdef{SE}[FOR]{For}{EndFor}[1]
  {\printandpush\algorithmicfor\ #1\ \algorithmicdo}
  {\popandprint\algorithmicend\ \algorithmicfor}%
\algdef{SE}[FOR]{ForAll}{EndFor}[1]
  {\printandpush\algorithmicfor\ #1\ \algorithmicdo}
  {\popandprint\algorithmicend\ \algorithmicfor}%
\algdef{SE}[LOOP]{Loop}{EndLoop}
  {\printandpush\algorithmicloop}
  {\popandprint\algorithmicend\ \algorithmicloop}%
\algdef{SE}[REPEAT]{Repeat}{Until}
  {\printandpush\algorithmicrepeat}[1]
  {\popandprint\algorithmicuntil\ #1}%
\algdef{SE}[IF]{If}{EndIf}[1]
  {\printandpush\algorithmicif\ #1\ \algorithmicthen}
  {\popandprint\algorithmicend\ \algorithmicif}%
\algdef{C}[IF]{IF}{ElsIf}[1]
  {\popandprint\pushindent\algorithmicelse\ \algorithmicif\ #1\ \algorithmicthen}%
\algdef{Ce}[ELSE]{IF}{Else}{EndIf}
  {\popandprint\pushindent\algorithmicelse}%
\algdef{SE}[FUNCTION]{Function}{EndFunction}[2]
   {\printandpush\algorithmicfunction\ \textproc{#1}\ifthenelse{\equal{#2}{}}{}{(#2)}}%
   {\popandprint\algorithmicend\ \algorithmicfunction}%
\makeatother

\algrenewcommand\algorithmicindent{1.0em}

\algnewcommand\algorithmicprint{\textbf{print }}
\algnewcommand{\Print}[1]{\algorithmicprint #1}

\algnewcommand\algorithmicand{\textbf{and}}
\algnewcommand{\AND}{\algorithmicand}

\algnewcommand\algorithmicswitch{\textbf{switch}}
\algnewcommand\algorithmiccase{\textbf{case}}
\algnewcommand\algorithmicassert{\texttt{assert}}
\algnewcommand\Assert[1]{\State \algorithmicassert(#1)}%
\algdef{SE}[SWITCH]{Switch}{EndSwitch}[1]{\algorithmicswitch\ #1\ \algorithmicdo}{\algorithmicend\ \algorithmicswitch}%
\algdef{SE}[CASE]{Case}{EndCase}[1]{\algorithmiccase\ #1}{\algorithmicend\ \algorithmiccase}%
\algtext*{EndSwitch}%
\algtext*{EndCase}%

\algrenewcommand\Return{\State \algorithmicreturn{} }%

\newcommand*{\tsf}[1]{
  \textsf{#1}
}

\newcommand*{\tbf}[1]{
  \textbf{#1}
}

\newcommand{\LAND}{\textbf{ \& }}

\newcommand*{\LargerCdot}{\raisebox{-0.25ex}{\scalebox{1.2}{$\cdot$}}}

\newcommand*{\defeq}{\stackrel{\text{def}}{=}}

\usetikzlibrary{calc,decorations.pathmorphing,shapes,arrows}
\newcommand\xxrsquigarrow[1]{%
\mathrel{%
\begin{tikzpicture}[baseline= {( $ (current bounding box.south) + (0,-0.5ex) $ )},decoration={zigzag,amplitude=0.7pt,segment length=1.2mm,post=lineto,
    post   length=4pt}]
  \node[inner sep=.5ex] (a) {$\scriptstyle #1$};
  \path[draw,->,decorate] 
    (a.south west) -- (a.south east);
\end{tikzpicture}}%
}

\makeatother

\usepackage{babel}
\begin{document}
\title{Optimal Morse functions and $H(\mathcal{M}^2,\mathbb{A})$ in $\tilde{O}(N)$ time}
\author[1]{Abhishek Rathore}   
\affil[1]{Visualization \& Graphics Lab., CSA Dept., Indian Institute of Science,  Bangalore, India.} 
\date{\vspace*{-2em}} 
\maketitle

\begin{abstract}
In this work, we design a nearly linear time discrete Morse theory based algorithm for computing homology groups of 2-manifolds, thereby establishing the fact that computing homology groups of 2-manifolds is remarkably easy. Unlike previous algorithms of similar flavor, our method works with coefficients from arbitrary abelian groups. Another advantage of our method lies in the fact that our algorithm actually elucidates the topological reason that makes computation on 2-manifolds easy. This is made possible owing to a new simple homotopy based construct that is referred to as \emph{expansion frames}. To being with we obtain an optimal discrete gradient vector field using expansion frames. This is followed by a pseudo-linear time dynamic programming based computation of discrete Morse boundary operator. The efficient design of optimal gradient vector field followed by fast computation of boundary operator affords us near linearity in computation of homology groups. \\
Moreover, we define a new criterion for nearly optimal Morse functions called pseudo-optimality. A Morse function is pseudo-optimal if we can obtain an optimal Morse function from it, simply by means of critical cell cancellations. Using expansion frames, we establish the surprising fact that an arbitrary discrete Morse function on 2-manifolds is pseudo-optimal. 
\end{abstract} 

Classical Morse Theory \cite{Mat02,Mi62} analyzes the topology of
the Riemannian manifolds by studying critical points of smooth functions
defined on it. In the 90's Robin Forman formulated a completely combinatorial
analogue of Morse theory, now known as discrete Morse theory. The
fact that Forman's theory can be formulated in language of graph theory
makes it possible to use powerful machinery from modern algorithmics
to provide efficient algorithms with rigorous guarantees. It is worth
noting that the reader can understand this work without any prior
knowledge of Morse theory as long as he understands the equivalent
graph theory problem. Knowledge of discrete Morse theory is however
useful for the more inclined reader who wishes to understand the context
and wider range of applicability of this work. In \autoref{sub:Graph-Theoretic-Reformulation},
we provide a quick overview of the graph theory setting of discrete
Morse theory in order to enable the reader to make a quick foray into
the core computer science problem at hand.

\section{Background and Preliminaries}

\subsection{Discrete Morse theory\label{sub:Discrete-Morse-Theory}}

Forman provides an extremely readable introduction to discrete Morse
theory in \cite{Fo02a}.

\begin{notation} The relation \textbf{'$\prec$'} is used to denote the following: $\tau \prec \sigma  \rightarrow  \tau \subset \sigma \,\, \& \, \dim \, \tau =  \dim\sigma - 1  $. 
\end{notation}

\begin{notation}[The d-(d-1) level of Hasse graph] By the term, d-(d-1) level of Hasse graph $\mathcal{H}$ we mean the subset of edges of the Hasse graph that join d-dimensional cofaces  to (d-1)-dimensional faces  of Hasse graph. 
\end{notation}

\begin{definition} \textbf{Boundary \& Couboundary of a simplex} $\sigma$:  We define the boundary and respectively coboundary of a simplex as $bd \: \sigma \,= \: \{ \tau \, | \, \tau \prec \sigma \}$ $cbd \: \sigma = \{ \rho \, | \, \sigma \prec \rho \}$ \end{definition}   

\begin{definition}\textbf{ Discrete Morse Function:} Let $\mathcal{K}$ denote a finite regular cell complex and let $\mathcal{L}$ denote the set of cells of $\mathcal{K}$. A function  $\mathcal{F}:\mathcal{L} \rightarrow\mathbb{R}$ is called a \textbf{discrete Morse function} (DMF) if it usually assigns higher values to higher dimensional cells, with at most one exception locally at each cell. Equivalently, a function $\mathcal{F}:\mathcal{L} \rightarrow\mathbb{R}$ is a discrete Morse function if for every $\sigma^{m} \in \mathcal{L}  $ we have:\\ 
(\textsf{A}.)\; $\mathcal{N}_1(\sigma)= \# \{\rho \in cbd \: \sigma | \mathcal{F}(\rho) \leq \mathcal{F}(\sigma) \} \leq 1 $\\ 
(\textsf{B}.)\; $\mathcal{N}_2(\sigma)= \# \{\tau \in bd \: \sigma \:\: | \mathcal{F}(\tau) \geq \mathcal{F}(\sigma) \} \leq 1 $\\
A cell $\sigma$ is \textbf{critical} if $\mathcal{N}_1(\sigma) = \mathcal{N}_2(\sigma)= 0$; A non-critical cell is a \textbf{regular cell}.  
\end{definition} 

\begin{definition} [Combinatorial Vector Field]   A \textbf{combinatorial vector field} (DVF) $\mathcal{V}$ on $\mathcal{L}$ is a collection of pairs of cells $\{\langle\alpha,\beta\rangle\}$ such that $\{\alpha^{m}\prec\beta^{(m+1)}\}$ and each cell occurs in at most one such pair of $\mathcal{V}$. 
\end{definition}  

 \begin{definition}[Discrete Gradient Vector Field]A pair of cells $\{\alpha^{m}\prec\beta^{(m+1)}\}$ s.t. $\mathcal{F}(\alpha) \geq \mathcal{F}(\beta)$ determines a \textbf{gradient pair}. A \textbf{discrete gradient vector field} (DGVF) $\mathcal{V}$ corresponding to a \textsf{DMF} $\mathcal{F}$ is a collection of cell pairs ${\alpha^{(p)}\prec\beta^{(p+1)}}$ such that ${\alpha^{(p)}\prec\beta^{(p+1)}}\in \mathcal{V}$ iff $\mathcal{F}(\beta)\leq \mathcal{F}(\alpha)$.  
\end{definition}

\begin{definition}
We define $\mathcal{V}$-path to be a cell sequence $\sigma_{0}^{(m)}$, $\tau_{0}^{(m+1)}$, $\sigma_{1}^{(m)}$, $\tau_{1}^{(m+1)}$, $\ldots\sigma_{q}^{(m)}$, $\tau_{q}^{(m+1)}$, $\sigma_{q+1}^{(m)}$ s.t. for $i=0,\ldots q,\{\sigma_{i}\prec\tau_{i}\}\in \mathcal{V}$, $\sigma_{i}\prec\tau_{i}\succ\sigma_{i+1}$ and $\sigma_{i}\neq\sigma_{i+1}$. The $\mathcal{V}$-path corresponding to a DMF $\mathcal{F}$ is a \textbf{gradient path} of $\mathcal{F}$. 
\end{definition}

\begin{theorem}[Forman~\cite{Fo98a}] \label{thm:weakmorse} Let $\mathcal{K}$ be a CW Complex with a \textsf{DMF} $\mathcal{F}$ defined on it. Then $\mathcal{K}$ is homotopy equivalent to a CW complex $\Omega$, such that $\Omega$ has precisely one m-dimensional cell for every m-dimensional critical cell in $\mathcal{K}$ and no other cells besides these. Moreover, let $c_{m}$ be the number of m-dimensional critical cells, $\beta_{m}$ the $m^{th}$ Betti Number w.r.t. some vector field $\mathcal{V}$ and $n$ the maximum dimension of $\mathcal{K}$. Then we have:\\
\textbf{The Weak Morse Inequalities:}
 \begin{align} 
(\textsf{A}.) & \;\textrm{For every }m \in  \{0\ldots n \}\textrm{: we have }c_{m}\geq \beta_{m}.\\
(\textsf{B}.) & \;c_{0}- c_{1} \ldots + (-1)^{n}c_{n} = \beta_{0} - \beta_{1} \ldots + (-1)^{n}\beta_{n} = \chi(\mathcal{K})   \label{eq:weakmorse}
\end{align}
\textbf{The Strong Morse Inequalities:}  
\begin{equation}
\textrm{For every }m \in  [0,n]: c_{m}- c_{m-1}\ldots + (-1)^{m}c_{0} \geq \beta_{m} - \beta_{m-1}\ldots + (-1)^{m}\beta_{0}  
\end{equation}
\end{theorem} 

\begin{notation} We shall denote the sum of Betti numbers by the symbol $\varLambda$ and sum of number of critical cells by symbol $\varUpsilon$. In other words, \[\varLambda\defeq\sum\limits_{i=0}^{n}\beta_i \quad \quad \varUpsilon\defeq\sum\limits_{i=0}^{n} c_i\]. 
\end{notation}

\begin{notation}The symbol $\tilde{O}(n)$ is used to indicate \emph{nearly linear}. It is given by $\tilde{O}(n)=n\left(\log n\right)^{O(1)}$.
\end{notation}

\begin{shnote} Given a DGVF, we can use topological sort to obtain a total order on the cells and then assign (arbitrary) ascending function values to the sorted list of cells. This will give us a Morse function that agrees with the partial order imposed by the gradient vector field. Any such Morse function will have the same critical cells as the gradient vector field. Hence, we shall use the terms \emph{optimal Morse function} and \emph{optimal gradient vector field} interchangeably.  
\end{shnote}

\begin{definition}[WMOC]\label{def:WMOC} Let $\varUpsilon(\mathcal{M})$ denote the sum of Morse numbers across all dimensions for the \textbf{optimal} DGVF on $\mathcal{M}$. We say that a family of simplicial complexes $\varOmega$ satisfies the weak Morse optimality condition (\textbf{WMOC}) when $\forall\mathcal{M}\in\varOmega$, $\varUpsilon(\mathcal{M}) = \tilde{O}(1)$. In other words, $\varUpsilon(\mathcal{M})\ll |\mathcal{M}|$ uniformly $\forall \mathcal{M} \in \varOmega$. 
\end{definition}

\subsection{Graph Theoretic Reformulation\label{sub:Graph-Theoretic-Reformulation}}

Given a simplicial complex $\mathcal{K}$, we construct its Hasse Graph representation $\mathcal{H_{\mathcal{K}}}$ (an undirected, multipartite graph) as follows: To every simplex $\sigma_{\mathcal{K}}^{d}\in\mathcal{K}$ associate a vertex $\sigma_{\mathcal{H}}^{d}\in\mathcal{H_{\mathcal{K}}}$. The dimension d of the simplex $\sigma_{\mathcal{K}}^{d}$ determines the \emph{vertex level} of the vertex $\sigma_{\mathcal{H}}^{d}$ in $\mathcal{H_{\mathcal{K}}}$. Every face incidence $(\tau_{\mathcal{K}}^{d-1},\sigma_{\mathcal{K}}^{d})$ determines an undirected edge $\langle\tau_{\mathcal{H}}^{d-1},\sigma_{\mathcal{H}}^{d}\rangle$ in $\mathcal{H_{\mathcal{K}}}$. Now orient the graph $\mathcal{H_{\mathcal{K}}}$ to a form a new directed graph $\mathcal{\overline{H_{\mathcal{K}}}}$. Initally all edges of $\mathcal{\overline{H}_{\mathcal{K}}}$ have default orientation. The default orientation is a directed edge $\sigma_{\mathcal{H}}^{d} \rightarrow\tau_{\mathcal{H}}^{d-1}\in\mathcal{\overline{H_{\mathcal{K}}}}$ that connects a k-dim node $\sigma_{\mathcal{H}}^{d}$ to a (k-1)-dim node $\tau_{\mathcal{H}}^{d-1}$. Finally, associate a matching $\mathcal{M}$ to graph $\mathcal{H_{\mathcal{K}}}$. If an edge $\langle\tau_{\mathcal{H}}^{d-1},\sigma_{\mathcal{H}}^{d}\rangle\in\mathcal{M}$ then, \emph{reverse} the orientation of that edge to $\tau_{\mathcal{H}}^{d-1} \rightarrow\sigma_{\mathcal{H}}^{d}\in\mathcal{\overline{H}_{\mathcal{K}}}$\textbf{.} The matching induced reorientation needs to be such that the graph $\overline{\mathcal{H}_{\mathcal{K}}}$ is a Directed Acyclic Graph. A graph matching on $\mathcal{H_{\mathcal{K}}}$ that leaves the graph $\overline{\mathcal{H_{\mathcal{K}}}}$ acyclic in the manner prescribed above is known as \emph{Morse Matching}.  Table~\ref{tab:dictionary} provides a translating dictionary from simplicial complexes to their Hasse graphs. See~\autoref{fig:HasseGraph}.    

\begin{figure}
\begin{minipage}[t]{1\columnwidth}%
\includegraphics[scale=0.05]{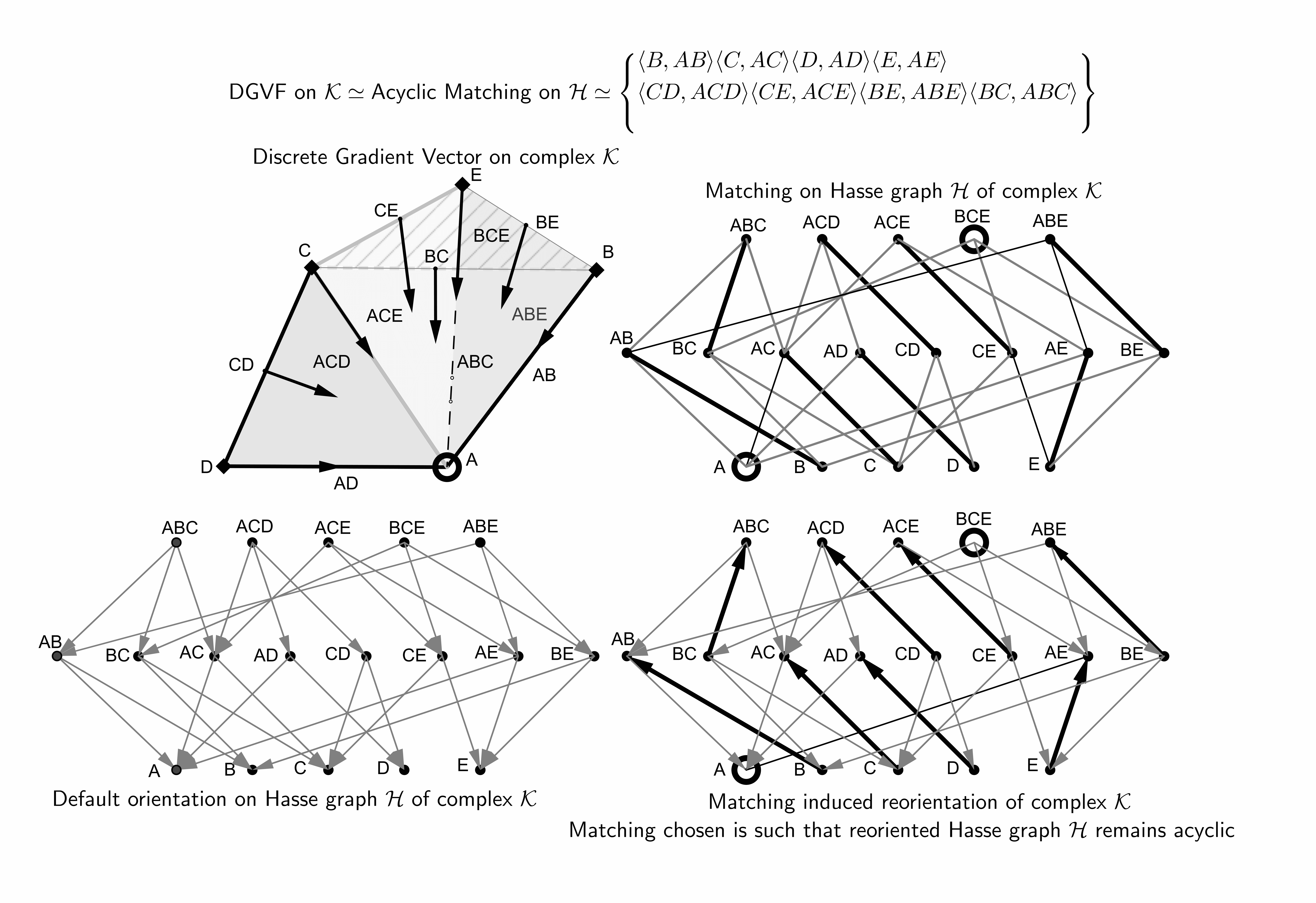}

\captionof{figure} {Matching induced orientation of Hasse Graph}

\label{fig:HasseGraph}%
\end{minipage}

\captionof{table}{Graph Theoretic dictionary for Morse Matching}

\centering{}\begin{center}
 \begin{center}  \resizebox{.75\textwidth}{!}{%
\begin{tabular}{|c||c|c|} \hline   & \textsf{Morse theory on Cell Complex $\mathcal{K}$} & \textsf{Graph theory on Hasse Graph $\overline{\mathcal{H}_{\mathcal{K}}}$}\tabularnewline \hline  \hline  
\textsf{1.} & \textsf{gradient Pair $\langle\alpha^{d-1},\beta^{d}\rangle\in\mathbb{V}$ } & \textsf{Matched pair of vertices $ (\alpha,\beta )\in\mathcal{H_{\mathcal{K}}}$ }\tabularnewline \hline  
\textsf{2.} & \textsf{Dimension d } & \textsf{Multipartite Graph Level d }\tabularnewline \hline  
\textsf{3.} & \textsf{$\sigma^{d-1}\prec\tau^{d}$ s.t. $\langle\sigma^{d-1},\tau^{d}\rangle\mathcal{\notin\mathbb{V}}$ } & \textsf{Default down-edge $\ensuremath{\tau} \rightarrow\ensuremath{\sigma}$}\tabularnewline \hline  
\textsf{4.} & \textsf{$\sigma^{d-1}\prec\tau^{d}$ s.t. $\langle\sigma^{d-1},\tau^{d}\rangle\mathcal{\in\mathbb{V}}$ } & \textsf{Matching up-edge $\sigma \rightarrow\tau$ }\tabularnewline \hline  
\textsf{5.} & $\mathbb{V}$\textsf{-Path} & \textsf{Directed Path}\tabularnewline \hline  
\textsf{6.} & \textsf{Non-trivial Closed }$\mathbb{V}$\textsf{-Path} & \textsf{Directed Cycle}\tabularnewline \hline  
\textsf{7.} & \textsf{CVF} & \textsf{Matching on the Hasse Graph}\tabularnewline \hline  
\textsf{8.} & \textsf{DGVF} & \textsf{Morse Matching (i.e. Acyclic Matching)}\tabularnewline \hline  
\textsf{9.} & \textsf{Critical Cell $\zeta^{d}$ } & \textsf{Unmatched Vertex $\zeta$}\tabularnewline \hline   
\textsf{10.} & \textsf{Regular Cell $\xi^{d}$ } & \textsf{Matched Vertex $\xi$}\tabularnewline \hline  
\end{tabular}%
}
\par\end{center}%
\par\end{center}\label{tab:dictionary}
\end{figure}

\subsection{Prior Work}

Joswig et al. \cite{JP06} proved the NP-completess of the decision
problem and posed the approximability of optimality of Morse gradient
vector fields (for general dimensional complexes) as an open problem,
by pointing out an error in Lewiner's claim about inapproximability
in \cite{Le02}. Recently \cite{Ra15} provided an $O(\log^{2}n)$
factor $\tilde{O}(n)$ time approximation algorithm for the optimal
discrete gradient vector field (that minimizes the number of critical
cells). Recently, Burton et al. \cite{BLPS13} developed an FPT algorithm
for optimizing Morse functions. Some of the notable works that seek
optimality of Morse matchings by applying heuristics in general are
\cite{AFV11,AFV12,HMMNWJD10,He05,JP06,Le03a,LB14}. The works that
constitute more relevant prior work for us are those that achieve
optimality by restricting the problem to 2-manifolds in nearly linear
time \cite{JM09,Le03b} and quadratic time \cite{BLW11} respectively.
Ours is however the first algorithm to compute homology groups of
2-manifolds \emph{with arbitrary coefficients} in nearly linear time.

\section{Boundary Operator Computation}

The analytic formula for boundary operator is given in Forman\cite{Fo98a}.
The obvious interpretation of the formula gives an exponential time
algorithm. We give an efficient $O(\kappa n)$ time algorithm where
$\kappa$ is the total number of critical cells which is nearly-linear
if the number of topologically interesting features are small relative
to the number of simplices. Hence, we have a \textit{pseudolinear
time complexity} algorithm for boundary operator Computation. We note
of the following Theorem from Forman\cite{Fo98a}: 

\begin{theorem} [Boundary Operator Computation. Forman~\cite{Fo98a}] \label{thm:bdryOp}  Consider an oriented simplicial complex. Then for any critical (p+1)-simplex $\beta$ set: \\
\[\triangle \beta = \sum\limits_{critical \: \alpha(p)}\: P_{\alpha\beta}\:\alpha \]
\[ P_{\alpha \beta} = \sum\limits_{\gamma \in \Gamma(\beta,\alpha)} \varTheta(\gamma) \] 
where $\Gamma(\beta,\alpha)$ is the set of discrete gradient paths which go from a face in the boundary of $\beta$ to $\alpha$  The multiplicity $\varTheta(\gamma)$ of any gradient path $\gamma$ is equal to $\pm1$ depending on whether given $\gamma$ the orientation on $\beta$ induces the chosen orientation on $\alpha$ or the opposite orientation. The formula for the boundary operator above computes the homology of complex K. 
\end{theorem}

We observe that we need 'formal sums' of critical cells at each critical
cell. However, there is an advantage in calculating these formal sums
for intermediate regular cells as well since this can potentially
speedup calculations at critical cells. Since topological sort also
does ordering for us, we can start at the lowest valued critical cell.
We proceed to the next higher valued cell and observe that we have
two cases. 

Also, we assume that our complex has a pre-assigned orientation. The
angular brackets $<,>$ in the formulae above denote the pre-assigned
orientation. Once the boundary operator is ready we use Smith normal
form algorithm over a collapsed complex that is provably significantly
smaller than the original complex, in a mathematically precise sense.

Let us denote by $\triangle\sigma$ the boundary operator computation
for cell $\sigma^{m}.$ We now make an inductive hypothesis that the
computation of the $\triangle$ operator has been done for all the
maximal faces / single coface (since they are all lower valued Morse
cells). Then the value of the $\triangle$ operator for the new cell
is calculated as follows:

\textbf{Case 1:} All flow emanating from a cell goes out through
its boundary faces. No lower-valued co-faces. 

\fbox{\begin{minipage}[t]{1\columnwidth}%
\begin{equation}
\triangle\sigma=\sum_{\substack{\tau_i\prec\sigma\\ \exists\xi\textsf{ s.t. } \langle \tau_i,\xi \rangle \in\mathscr{V}_{m}\\  \langle \tau_i,\sigma \rangle \notin\mathscr{V}_{m} } }\triangle\tau_i\times<\partial\sigma,\tau_i>+\sum_{\substack{\bm{\alpha_j}\prec\sigma\\ \cancel{\exists}\xi\textsf{ s.t. } \langle \bm{\alpha_j},\xi \rangle \in\mathscr{V}_{m}\\ \cancel{\exists}\zeta\textsf{ s.t. } \langle \zeta,\bm{\alpha_j} \rangle \in\mathscr{V}_{m-1} } }\bm{\alpha_j}\times \langle\partial\sigma,\bm{\alpha_j} \rangle 
\end{equation}%
\end{minipage}}

The first formula takes care of Case 1 where flow goes out through
the faces of the boundary. Note that in the formula above, $\tau^{m-1}$
is a placeholder for non-critical faces (if any) of $\sigma^{m}$,
i.e. $\{\tau^{m-1}\prec\sigma^{m}\}$, which are not a part of the
Discrete Gradient vector field which is equivalent to saying $f(\tau)<f(\sigma)$.
Similarly, $\alpha^{m-1}$ is a representative for the critical faces
(if any) of$\sigma^{m}.$ This formula holds irrespective of whether
$\sigma^{m}$ itself is critical or non-critical. In case of computation
of boundary of a critical $\sigma^{m}$ such that $m=0$, i.e. when
$\sigma^{m}$ is a critical point, the boundary is null.

\textbf{Case 2:} The cell has $1$ lower valued coface. 

\fbox{\begin{minipage}[t]{1\columnwidth}%
\begin{equation}\triangle \sigma = <\partial \beta, \sigma> \times \triangle \beta  \end{equation}     %
\end{minipage}}

The second \textcolor{black}{formula} takes care of Case 2 when $\sigma^{m}$
has a lower valued co-face $\beta^{m+1}$. 

\textbf{Case 3:} The 0-dimensional cell $\sigma$ is the unique minima. 

\fbox{\begin{minipage}[t]{1\columnwidth}%
\begin{equation}\triangle \sigma = \emptyset  \end{equation}    %
\end{minipage}}

\begin{theorem}[Boundary Operator Computation: Correctness Proof] The Algorithm correctly computes boundary operator $\triangle$. \label{thm:bdryOpCorr}
\end{theorem}
\begin{proof} Note that, to begin with we start with a list of cells in an ascending total order. Let us call this list $\mathcal{L}$. This total order is one of the total orders that is compatible with the partial order prescribed by the gradient vector field $\mathscr{V}$. If we assign the function value 'i' i.e. the index of some cell $\mathcal{L}[i]$ to each cell in $\mathcal{L}$, we essentially obtain a Morse function compatible with the gradient vector field. The first cell we process is one with the lowest function value (i.e. the unique minima). This cell is then followed by cells with increasingly higher Morse function values. To prove that the formulaic computation of the $\triangle$ operator as expressed in subroutine \textbf{calcBdryOp}() is, in fact, the same as expressed in Theorem \ref{thm:bdryOp} we proceed by induction. Let $\sigma_1$ denote the unique minima. The base case of induction for $\triangle \sigma_1$ is trivial. Now suppose that for all cells in the set $ \{\sigma_1, \sigma_2, \ldots\sigma_I \}$, we have correctly computed the boundary operator as prescribed in Theorem \ref{thm:bdryOp}. Now suppose we encounter cell $\sigma_{I+1}$. Suppose that $\sigma_{I+1}$ has a lower valued coface $\beta$ i.e. ($\sigma_{I+1}\prec\beta$ \LAND $ \langle\sigma_{I+1}, \beta \rangle\in\mathscr{V}$). Since $\beta$ has lower function value as compared to $\sigma_{I+1}$ (by hypothesis), we conclude that $\beta=\sigma_{J+1}$ for some $J<I$.
All paths emanating from $\sigma_{I+1}$ must go through $\beta$. The orientation induced by some path $\gamma_i$ $\beta\overset{\gamma_{i}}{ \rightsquigarrow}\rho$  from $\beta$ to some critical cell say $\rho$ is $\iota$ where $\iota=\pm 1$, then the orientation of path  $\sigma_{I+1}\xxrsquigarrow{\beta\circ\gamma_{i}}\rho$ will be $  \langle\partial \beta, \sigma_{I+1}   \rangle \times \iota$. Therefore, the total count of paths (with induced orientation accounted for) will be $  \langle \partial \beta, \sigma   \rangle \times \triangle \beta$. Hence, the boundary operator computation done in $\textbf{calcBdryOp}()$ is valid for the case when $\sigma_{I+1}$ has a lower valued coface. \\
Finally, assume that $\sigma_{I+1}$ does not have any lower valued coface. Therefore, the flow leaving from $\sigma_{I+1}$ will be through each of its faces (except possibly one higher valued face). If it indeed has a (matched) higher valued face then flow will be \emph{entering it} through that face and hence the face in question isn't relevant in calculating the weighted sum of gradient paths that \emph{leave} $\sigma_{I+1}$. When consider lower valued faces of $\sigma_{I+1}$, we make a distinction between faces that are non-critical and those those that are critical. If a face say $\bm{\alpha_j}$ is critical, then clearly we are justified in directly including the entry $\bm{\alpha_j}\times \langle\partial\sigma,\bm{\alpha_j} \rangle$ as part of our «formal sum» that makes up the cell boundary. As for the non-critical entries of the formula, namely $[\triangle\tau_i\times<\partial\sigma,\tau_i>]$, we impose an additional constraint  $ \langle \tau_i,\xi \rangle \in\mathscr{V}_{m}$ (as opposed to $ \langle \xi,\tau_i \rangle \in\mathscr{V}_{m-1}$) in the summation. In doing so, we are ruling out all entries that would valid directed paths going out of $\sigma_{I+1}$ but those that won't add up to make gradient paths as prescribed by Theorem \ref{thm:bdryOp}. Now since $\tau_i$ is lower valued its boundary $\triangle\tau_i$ has already been calculated correctly by Induction Hypothesis. But clearly every gradient path emerging from $\sigma_{I+1}$ must first pass through one of these $\tau_i$'s.  Also, for each of these gradient paths, the orientations will change precisely by the multiple of $ \langle \tau_i,\xi \rangle$. Therefore the weighted sum of (non-trivial) gradient paths from $\sigma_{I+1}$ will be the sum of all the contributions by boundaries of each of the non-critical faces $\tau_i$. To complete the argument for the induction step, we note that these sums along with contributions from the critical faces of $\sigma_{I+1}$ takes into account each gradient path precisely once. Also, it is easy to see that multiplication by co-orientation at each step provides the weights to ensure that the final entry will decide the induced orientation. Hence proved.
\end{proof}

\begin{theorem}[Complexity of Computing Boundary Operator]
The complexity of computing the boundary operator is $O(\varUpsilon \times \mathcal{N})$.
\end{theorem} 
\begin{proof}
For the Hasse graph $\mathcal{H}(\mathcal{V},\mathcal{E})$ of  a simplicial complex, $\mathcal{E}\leq\mathcal{V}\times \mathcal{D}$ where $\mathcal{D}$ is the maximum dimension of cells in the complex (which in our case is 2). Therefore, $|\mathcal{E}|=O(|\mathcal{V}|)$. (It is easy to show that for a cubical complexes as well, number of edges is $O(|\mathcal{V}|)$).The complexity of computing topological sort of the oriented Hasse graph  is $O(|\mathcal{V}|+|\mathcal{E}|)$ which is same as $O(|\mathcal{V}|)$, assuming that our input manifold is either simplicial or cubical. \\
 The \textbf{for} loop in Lines \ref{lst:line:loopBS}-\ref{lst:line:loopBE} of procedure \textbf{calcBdryOp}() costs at the most $O(\varLambda)$ per iteration while the total number of iterations is $O(|\mathcal{V}|)$. But since in our case, using \autoref{thm:equality}, $\varLambda=\varUpsilon$, the total cost of the for loop is $O(|\mathcal{V}|\times\varUpsilon)$. Therefore, complexity of computing boundary operator is $O(|\mathcal{V}|\times\varUpsilon)=O(\mathcal{N}\times\varUpsilon)$, since the number of vertices in the Hasse graph is same as number of cells in the complex (i.e. the size of the complex namely $\mathcal{N}$).
\end{proof}

It is worth noting that in vast majority of the practical scenarios
$\mathcal{N}\ggg\varUpsilon$, enough for us to assume that compared
to the size of the complex, the 'topological complexity', $\varUpsilon$
is nearly a constant. We therefore use the notation $\tilde{O}(\LargerCdot)$
(where $\ensuremath{O(\varUpsilon\times\mathcal{N})}=\tilde{O}(\mathcal{N})$)
to indicate the \emph{nearly linear time} complexity of boundary operator
computation.

\section{Frames of Expansion \label{sec:Frames-of-Expansion}}

\subsection{Basic Formulation}

\begin{notation}
Let the set $\mathcal{B}(\alpha)$ denote the 0-dim cells (vertices) and 1-dim. cells (edges) incident on $\alpha$ if $\alpha$ is 2-dimensional and let $\mathcal{B}(\alpha)$ denote the 0-dim cells incident on $\alpha$ if $\alpha$ is 1-dimensional.
\end{notation}

\begin{definition}[Semigraph] A semigraph is a set of vertices and edges s.t. every edge may have either one or two vertices incident on it.
\end{definition}

Semigraphs generalize graphs in the sense that, in a graph, every
edge is incident on precisely two vertices. 

\begin{definition}[Frame of expansion of a critical cell] Given a critical cell $\alpha^{n}$ where ($n\geq 2$), consider the set of all cells that can be reached from $\alpha$, by following one of the gradient paths within the gradient vector field. We call this set the \emph{expansion set} of critical cell $\alpha$ and denote it by $\widehat{\alpha}$. The \emph{frame of expansion} of $\alpha$ is the $n-1$-dim. boundary of $\widehat{\alpha}$ along with the $n-2$ dim. cells incident on these boundary cells. We denote the frame of expansion of $\alpha$ by $\underline{\widehat{\alpha}}$. 
\end{definition}

\begin{definition}[Frame of expansion of a boundary cell] Given a regular boundary cell $\xi^{n-1}$ where ($n\geq 2$), suppose that $\langle\xi,\chi\rangle$ forms a gradient pair. Now consider the set of all cells that can be reached from $\xi$, by following one of the gradient paths within the gradient vector field. We call this set the \emph{expansion set} of boundary cell $\xi$ and denote it as $\widehat{\xi}$. The \emph{frame of expansion} of $\xi$ is the $n-1$-dimensional boundary of $\widehat{\xi}$. We denote it as $\underline{\widehat{\xi}}$. 
\end{definition}

\begin{shnote}[Method of addition of cells upon expansion] It must be noted that if there is an expansion along $\tau^{1}$ into cell $\varpi^{2}$, then we delete $\tau^{1}$ from the frame and the set $\mathcal{B}(\varpi\setminus\tau)$ is added into the frame.
\end{shnote}

\begin{shnote}\label{note:2paths} Suppose we are given a regular 2-cell $\varpi^{2}$, s.t. the 1-cell $\tau^{1}\in\varpi^{2}$. The boundary of $\tau$ namely $\mathcal{B}(\tau)$ consists of two vertices say $\lambda^{0}$ and $\rho^{0}$. Note that within the set $\mathcal{B}(\varpi)$, there exist two non-intersecting paths that connect $\lambda$ and $\rho$. One path involves the singular edge $\tau$, the other path consists of edges belonging to the set $\mathcal{B}(\varpi\setminus\tau)$    
\end{shnote}

\begin{definition}[connectedness, connecting path] Consider two cells $\sigma^{(m-1)},\tau^{(m-1)}$ in a complex $\mathcal{K}^{n}$. We say that $\sigma$ and $\tau$ are said to be \textbf{Type 1 connected} in complex $\mathcal{K}$ if there exists a cell sequence $\phi_0^{(m-1)}$, $\gamma_{0}^{(m)}$, $\phi_{1}^{(m-1)}$, $\gamma_{1}^{(m)}$, $\ldots\phi_{q}^{(m-1)}$, $\gamma_{q}^{(m)}$, $\phi_{q+1}^{(m-1)}$ s.t. for $i=0,\ldots q$, $\phi_{i}\prec\gamma_{i}\succ\phi_{i+1}$, $\phi_0=\sigma$, $\phi_{q+1}=\tau$ and $\phi_{i}\neq\phi_{i+1}$. This sequence of cells, $\phi_0^{(m-1)}\ldots\phi_{q+1}^{(m-1)}$ is known as a \textbf{connecting path}. Analogously, we say that $\sigma^{m}$ and $\tau^{m}$ are \textbf{Type 2 connected} in complex $\mathcal{K}^{n}$ if there exists a cell sequence $\gamma_0^{(m)}$, $\phi_{0}^{(m-1)}$, $\gamma_{1}^{(m)}$, $\phi_{1}^{(m-1)}$, $\ldots\gamma_{q}^{(m)}$, $\phi_{q}^{(m-1)}$, $\gamma_{q+1}^{(m)}$ s.t. for $i=0,\ldots q$, $\gamma_{i}\succ\phi_{i}\prec\gamma_{i+1}$, $\gamma_0=\sigma$, $\gamma_{q+1}=\tau$ and $\gamma_{i}\neq\gamma_{i+1}$. The sequence of cells, $\gamma_0^{(m)}\ldots\gamma_{q+1}^{(m)}$ is known as a \textbf{connecting path}. Finally, we say that, $\sigma^{m}$ and $\tau^{(m-1)}$ are \textbf{Type 3 connected} if there exists a cell $\eta^{(m-1)}\prec\sigma^m$ with Type 2 connectedness between $\eta$ and $\tau$.     

Finally, we say that a set of $m-1$ and $m$ dim. cells are said to form a \textbf{connected set} if for any pair of  $m-1$ dim. cells (alternatively, for any pair of $m$ dim. cells) we can find sequence of connecting cells as prescribed above.
\end{definition}

\subsection{Pseudocode for $\tilde{O}(n)$-Time Algorithm for Computing Homology
of 2-manifolds}

\begin{notation} Given manifold $\mathcal{M}$, we use the notation $\mathcal{M}^d$ to denote the d-dimensional cells of manifold $\mathcal{M}$. 
\end{notation}

\begin{definition}[Boundary faces, Coboundary faces] Given a complex $\mathcal{K}$, if there exists a cell $\vartheta^{d}$ of dimension $d$ s.t. there exists a unique $(d+1)$-dimensional cell $\varpi^{(d+1)}$ satisfying $\vartheta\prec\varpi$, then we call $\vartheta$ a d-dimensional boundary face of complex $\mathcal{K}$. Also, in this case, $\varpi$ is known as a $(d+1)$-dimensional coboundary face of complex $\mathcal{K}$.
\end{definition}

\begin{definition}[Boundary and Coboundary] Given a complex $\mathcal{K}$, the list of all d-dimensional boundary faces of $\mathcal{K}$ is known as the d-dimensional boundary of $\mathcal{K}$. Also, list of all d-dimensional coboundary faces of $\mathcal{K}$ is known as the d-dimensional coboundary of $\mathcal{K}$.
\end{definition}

\begin{definition}[$n$-flow] The set of gradient paths in vector field $\mathscr{V}$ on manifold $\mathfrak{M}$ that involve alternating $n$-dim. and $(n-1)$-dim. cells is known as the n-flow of $\mathscr{V}$
\end{definition}

\begin{algorithm}
\begin{algorithmic}[1]   

\Procedure{calcHomology}{$\mathcal{M},\mathcal{A}$;}
\LState{We use the \textbf{mainFrame()} subroutine to design a vector field $\mathscr{V}$ on $\mathcal{M}$.}
\LState{We then use subroutine \textbf{calcBdryOp()} to calculate the boundary operator $\triangle_c$ for DGVF $\mathscr{V}$.}
\LState{Finally, using chain complex implied by boundary operator $\triangle_c$, we calculate homology of $\mathcal{M}$ (with coefficients coming from arbitrary abelian group $\mathbb{A}$) using Smith Normal Form.}
\EndProcedure 

\Statex

\Procedure{calcBdryOp}{$\mathcal{M},\mathcal{H},\mathscr{V}$}
\LState{$\tbf{topologicalSort}(\mathcal{H},\mathscr{V},\mathcal{L},\textsf{'ASCENDING')}$;}
\LState{$\sigma_1 =\mathcal{L}[1]$; $\triangle \sigma_1 = \emptyset $;}
\ForAll{$2 \leq  i \leq \abs{\mathcal{L}}$; $\sigma\colonequals\mathcal{L}[i]$}  	\label{lst:line:loopBS} 	 
	\If{$\langle\sigma, \beta \rangle$ is a gradient pair}
		\LState{$\triangle \sigma = <\partial \beta, \sigma> \times \triangle \beta $;}
	\Else
		\LState{Let $\tau_i \prec \sigma$ be the set of regular cells incident on $\sigma$ s.t. $ \langle\tau_i,\sigma \rangle\notin \mathscr{V}$;}
		\LState{Let $\alpha_i \prec \sigma$ be the set of critical cells incident on $\sigma$;}
		\LState{$\triangle \sigma = \sum \triangle \tau_i \times <\partial \sigma, \tau_i>  +  \sum \bm{\alpha_i} \times <\partial \sigma, \alpha_i> $;}
	\EndIf
	\If{$\sigma$ is a critical cell} 
		\LState{$\triangle_c \sigma\colonequals \triangle\sigma$;}
	\EndIf
\EndFor			 		\label{lst:line:loopBE} 
\LState{$\triangle_c$ is the Morse boundary operator corresponding to vector field $\mathscr{V}$.}	
\EndProcedure 
\end{algorithmic}

\protect\caption{$\textbf{Homology}()$ }

\label{alg:hom}
\end{algorithm}

\begin{algorithm}
\begin{algorithmic}[1]   

\Procedure{findCoBdry}{$\mathcal{M}^{d}$}
	\LState{Let $\mathcal{M}^{d}$ be the list of d-dimensional cells of manifold $\mathcal{M}$. Scan through the list $\mathcal{M}^{d}$. If a cell $\mathcal{M}^{d}[i]$ has a face $\vartheta$ such that $\vartheta$ is the sole coface of $\mathcal{M}^{d}[i]$ then add $\mathcal{M}^{d}[i]$ to $\mathcal{B}^d$. $\mathcal{B}^d =  \{ \varpi\in\mathcal{M}^{d}\,|\,\varpi \tsf{ has at least one boundary face. } \}$}
\LState{\tbf{return }$\mathcal{B}^d$;}
\EndProcedure 

\Statex

\Procedure{addPairToVectorField}{$\tau, \vartheta,\mathscr{V}, \mathcal{M},\mathcal{B}^d,d$}
\LState{If $\tau\neq$ NIL and if $\tau$ isn't already matched then do the following:}
\LState{\quad(a.) Match $\tau$ to $\vartheta$.\quad(b.) Delete $\vartheta$ and $\tau$ from $\mathcal{M}^{d-1}$ and $\mathcal{M}^{d}$ respectively.}
\LState{\quad(c.) If $\tau\in\mathcal{B}^{d}$ then delete $\tau$ from list $\mathcal{B}^{d}$.\quad(d.) Enqueue $\tau$ in $\mathcal{Q}$.}
\LState{\quad(e.) Add $ \langle\vartheta,\tau \rangle$ to vector field $\mathscr{V}$.}
\EndProcedure 

\Statex

\Procedure{frameFlow}{$\mathcal{M}^{d},d,\mathcal{B}^d, \mathscr{V}$}
	\LState{Dequeue a cell $\varpi$ from $\mathcal{B}^{d}$. If the dequeue operation with list $\mathcal{B}^{d}$ returns NIL then dequeue a cell from list $\mathcal{M}^{d}$.}
	\Repeat \label{lst:line:loopdisjoint1}
		\Repeat
			\If{$\varpi$ is a cell that has been dequeued from list $\mathcal{B}^{d}$} 
				\If{$\upsilon$ is a boundary face of $\varpi$} \label{lst:line:bdryadd1}
					\LState{Invoke \tbf{addPairToVectorField()} in order to add $\langle\upsilon,\varpi\rangle$ to $\mathscr{V}$.} \label{lst:line:bdryadd2}
				\Else	
					\LState{ Delete $\varpi$ from $\mathcal{M}^d $.}
				\EndIf 
			\EndIf
			\ForAll{each face $\vartheta_{i}$ of $\varpi$}  		 			 		
				\LState{If there exists $\mu_i\succ\vartheta_{i}$ s.t. $\mu_i$ isn't part of any gradient pair of $\mathscr{V}$} 
				\LState{Then invoke \tbf{addPairToVectorField()} to add $\langle\vartheta_i,\mu_i\rangle$ to $\mathscr{V}$.}
			\EndFor  
			\LState{Dequeue a cell from queue $\mathcal{Q}$. Call it $\varpi$.}
		\Until {$\varpi \neq \tsf{NIL}$} 
		\LState{Dequeue a cell from queue $\mathcal{B}^d$. Call it $\varpi$.} \label{lst:line:bdrydisjoint}
	\Until {$\varpi\neq\tsf{NIL}$}  \label{lst:line:loopdisjoint2}
\EndProcedure

\Statex

\Procedure{mainFrame}{$\mathcal{M}$}
	\LState{Invoke \tbf{findCoBdry()} to find coboundary $\mathcal{B}^2$ of $\mathcal{M}^2$.}
	\LState{Use \tbf{frameFlow()} to design vector field $\mathscr{V}$ on cells of $\mathcal{M}^{2}$.}
	\LState{$\mathcal{E}[1:\tsf{numEars}]\colonequals\tbf{earDecompose}(\mathscr{M}$)}
	\ForAll{$1 \leq  i \leq \tsf{numEars}$}
		\LState{ Invoke \tbf{findCoBdry($\mathcal{E}_i$)} to find $\mathcal{B}^1[i]$}
		\LState{Use \tbf{frameFlow()} to design vector field $\mathscr{V}$ on cells of ear $\mathcal{E}_i$.}
	\EndFor
	\LState{$\tbf{return }\mathscr{V};$}
\EndProcedure
\end{algorithmic}

\protect\caption{Frame Flow}

\label{alg:frameflow}
\end{algorithm}

\subsection{Frame Expansions: Correctness \& Complexity Proof}

\begin{lemma} \label{lem:staycon} Suppose there exist two vertices $\alpha^{0}$ and $\gamma^{0}$ that are connected through edges that belong to some frame after a certain number of elementary expansions. Then the two vertices will remain connected through edges belonging to that frame upon further expansions  
\end{lemma}
\begin{proof} By hypothesis, we assume that two vertices, say $\alpha^{0}$ and $\gamma^{0}$ are connected through edges belonging to the frame after a certain number of expansions. Therefore there exists a connecting path $\mathcal{P}$ connecting the two vertices. Suppose w.l.o.g., we expand along some edge $\tau^{1}$ into cell $\varpi^{2}$. We have two cases.\\
\textbf{Case 1:}$\tau^{1}\notin\mathcal{P}$. In this case, all edges in path $\mathcal{P}$ continue to belong to the frame after the expansion corresponding to gradient pair $\langle\tau,\varpi\rangle$. Therefore, even after this expansion, $\alpha$ and $\gamma$ remain connected. \\
\textbf{Case 2:}$\tau^{1}\in\mathcal{P}$. Suppose $\lambda^{0}$ and $\rho^{0}$ are the vertices of $\tau$. Then there exists a path $\mathcal{P}_1$ s.t. $\mathcal{P}_1\subset\mathcal{P}$ connecting $\alpha$ and $\lambda$. Also there exists another path $\mathcal{P}_2$ s.t. $\mathcal{P}_2\subset\mathcal{P}$ connecting $\rho$ and $\gamma$.  However, from \ref{note:2paths} we know that, $\lambda^{0}$ and $\rho^{0}$  are connected through edges that belong to set $\mathcal{B}(\varpi\setminus\tau)$. The original path $\mathcal{P}$ consists of edges $\mathcal{P}_1\cup\tau\cup\mathcal{P}_2$. Upon expansion, we have a new path namely $\mathcal{P}_1\cup\{\varpi\setminus\tau\}\cup\mathcal{P}_2$. Therefore, frame expansions maintain connectivity.         
\end{proof}

\begin{shnote}[2-Manifolds and Semi-graphs] \label{note:semigraph} A 2-manifold without boundary has the structure of a simple graph (unrelated to Hasse graphs) in the following sense: Let every 2-cell denote a vertex and let every 1-cell denote an edge connecting 2-cells. The manifold structure allows at most two incident 2-cells for every 1-cell, whereas not having a boundary implies that the incidence number is exactly two for every 1-cell. Now if we have a 2-manifold with boundary, then the boundary 1-cells will have only one incident 2-cell whereas all other 1-cells will have two incident 2-cells. Therefore a 2-manifold with boundary has the structure of a semi-graph. For a given 2-manifold $\mathcal{M}$, let us denote the semigraph structure by $\mathcal{G}_s(\mathcal{M})$.
\end{shnote} 

\begin{lemma} \label{lem:allcon} Every vertex belonging to manifold $\mathcal{M}$ is included in the frame, when all expansions are processed.
\end{lemma}
\begin{proof} From \autoref{note:semigraph}, we know that the 2-cells and 1-cells of a given 2-manifold $\mathcal{M}$ forms a semi-graph structure which we denote by $\mathcal{G}_s(\mathcal{M})$. We use the following convention: If a 2-cell, say $\tau$ is included in some gradient pair belonging to vector field $\mathscr{V}$ or if $\tau$ is the start cell of procedure \textbf{frameFlow()} described in \autoref{alg:frameflow}, the we say that vertex $\tau$ is traversed.\\

\textbf{Case 1:} Suppose $\mathcal{M}$ has no 1-dim. boundary faces. Then $\mathcal{G}_s(\mathcal{M})$ assumes the structure of a connected  simple graph. In this case, the procedure \textbf{frameFlow()} described in \autoref{alg:frameflow} we begin with some starting 2-cell $\varpi$. While scanning through all the faces of $\varpi$, if we find a face $\vartheta$ s.t. $\mu\neq\varpi$ and $\mu\succ\vartheta$ and $\mu$ isn't part of any gradient pair, then we \emph{traverse} $\mu$ by adding gradient pair $\vartheta,\mu$ to vector field $\mathscr{V}$ and add $\mu$ to a queue. Having processed all faces of $\varpi$, we dequeue a cell, say $\varpi_{\textsf{new}}$ from the the queue. We process $\varpi_{\textsf{new}}$ in exactly the same way as we process $\varpi$. And we keep doing this till the queue is empty. Clearly, this is equivalent to a breadth first traversal on graph $\mathcal{G}_s(\mathcal{M})$. Given the fact that all vertices of a graph are traversed in a breadth first traversal, we conclude that except for the start cell, all other 2-cells are part of some gradient pair. When the start cell $\varpi$ is added, the expansion frame consists of $\mathcal{B}(\varpi)$. Every time we add a gradient pair $\langle\vartheta,\mu\rangle$ to $\mathscr{V}$, we delete $\vartheta$ from the frame and add $\mathcal{B}(\mu\setminus\vartheta)$ to the frame. Since every vertex $v_i$ belonging to $\mathcal{M}$ is part of $\mathcal{B}(\mu\setminus\vartheta)$ for some 2-cell $\mu$, we see that each vertex $v_i$ becomes part of the expansion frame at some stage of the construction of the frame. When new gradient pairs are processed, we may delete 1-cells from out frame, but 0-cells are never deleted. So, all vertices of $\mathcal{M}$ eventually become part of the expansion frame. See \autoref{fig:hollowcube1} and \autoref{fig:hollowcube2} for an example.\\

\textbf{Case 2:} Suppose $\mathcal{M}$ has some 1-dim. boundary faces. In this case, $\mathcal{G}_s(\mathcal{M})$ has the structure of a possibly disconnected semigraph. \footnote{Example: For a connected manifold $\mathcal{K}$, this may happen, for instance when, say one the 2-cells $A$ is connected to other 2-cells only by the medium of 0-cells while the 1-cells of $A$ are not shared with other cells. See \autoref{fig:disjoint1}, \autoref{fig:disjoint2} and \autoref{fig:disjoint3} for another such example.} If the manifold has a coboundary face, say $B_1$ then for the first connected component of $\mathcal{G}_s(\mathcal{M})$ in lines \ref{lst:line:bdryadd1}-\ref{lst:line:bdryadd2} of Procedure \emph{frameflow()} in \autoref{alg:frameflow}, we add the boundary-coboundary pair to vector field $\mathscr{V}$. Following that, the cells that are in the same connected component are added to the vector field in a manner similar to \emph{Case 1}. If there exists another connected component, then surely such a connected  component must have at least one coboundary face. In lines \ref{lst:line:loopdisjoint2} and \ref{lst:line:bdrydisjoint} of Procedure \emph{frameflow()} in \autoref{alg:frameflow}, we check if such a coboundary face exists. If it does exist then in the loop \ref{lst:line:loopdisjoint1}-\ref{lst:line:loopdisjoint2}, we process the every connected component in the same way as we process the very first one. Given the fact that all connected components of semigraph $\mathcal{G}_s(\mathcal{M})$ are processed, every 2-cell $\varpi_{\textsf{new}}$ in each of these components is part of vector field $\mathscr{V}$. Suppose $\varpi_{\textsf{new}}$ is paired with some 1-cells $\vartheta_{\textsf{new}}$ each time then, each time we delete $\vartheta_{\textsf{new}}$ from the the frame and add $\mathcal{B}(\mu\setminus\vartheta_{\textsf{new}})$ to the frame. When new gradient pairs are processed, we may delete 1-cells from out frame, but 0-cells are never deleted. Since every vertex is incident on at least one of the 2-cells in one of the connected components, we establish the fact that all vertices eventually become part of the expansion frame. See \autoref{fig:disjoint1}, \autoref{fig:disjoint2} and \autoref{fig:disjoint3} for such an example.          
\end{proof}

\begin{shnote} Every connected component of $\mathcal{G}_s(\mathcal{M})$ has a 2-cell which shares a 0-cell with a 2-cell from another connected component. If we imagine every connected component of $\mathcal{G}_s(\mathcal{M})$ as a vertex and every shared 0-cell as a hyperedge then we get a connected hypergraph that we write as $\mathcal{H}_c(\mathcal{M})$. We know that $\mathcal{H}_c(\mathcal{M})$ is connected because if this were not the case then clearly $\mathcal{M}$ itself will have more than one connected components. We call $\mathcal{H}_c(\mathcal{M})$ the component hypergraph of $\mathcal{M}$. 
\end{shnote}

 Consider a gradient vector field $\mathcal{V}$ assigned to a manifold$\mathcal{M}$. First consider the case when $\mathcal{M}$ is a manifold without boundary. Consider a critical cell $\alpha^{2}$. Note that before we do any expansions, $\mathcal{B}(\alpha)$ is our original $\underline{\widehat{\alpha}}$. Let $\sigma\succ\tau\prec\alpha$. If $\langle \tau,\sigma \rangle\in \mathcal{V}$ then we can consider it as an expansion along $\tau$ to the cell $\sigma$. Now, as per the definition of frame expansion, we add the set $ \{\mathcal{B}(\sigma)\setminus\{\tau\} \}$ to $\underline{\widehat{\alpha}}$ and we delete $\{\tau\}$ from $\underline{\widehat{\alpha}}$. Therefore, $$\underline{\widehat{\alpha}}=\underline{\widehat{\alpha}}- \{\tau \}+ \{\mathcal{B}(\sigma)\setminus\{\tau\}\}$$ But, this is same as saying, $\underline{\widehat{\alpha}}= \{\mathcal{B}(\alpha)\setminus\{\tau\} \}+ \{\mathcal{B}(\sigma)\setminus\{\tau\} \}$. Clearly, the sets $ \{\mathcal{B}(\alpha)\setminus\{\tau\} \}$ and $ \{\mathcal{B}(\sigma)\setminus\{\tau\} \}$ are themselves connected and both these sets have a common boundary namely $\mathcal{B}(\tau)$ (the boundary of $\tau$). Therefore, expansion along $\tau$ preserves the connectivity of $\underline{\widehat{\alpha}}$. Now, given our intermediate stage $\underline{\widehat{\alpha}}$, if any of the 1 dim. cells say $\vartheta_{i}^{1}\in\underline{\widehat{\alpha}}$ forms a gradient pair with a $2$ dim. cell $\varpi_{i}^{2}$ then by expansion we have, $$\underline{\widehat{\alpha}}=\underline{\widehat{\alpha}}-\{\vartheta_{i}\}+\{\mathcal{B}(\varpi_{i})\setminus\{\vartheta_{i}\}\}$$        Each time we observe that the boundary of $\underline{\widehat{\alpha}}$ and the boundary of $ \{\mathcal{B}(\varpi_{i})\setminus\{\vartheta_{i}\}\}$ is, in fact, the same as the boundary of $\vartheta_{i}$ namely $\mathcal{B}(\vartheta_{i})$. Therefore, upon expanding the frame $\underline{\widehat{\alpha}}$ along $\vartheta_{i}$, connectivity of $\underline{\widehat{\alpha}}$ is preserved and $\underline{\widehat{\alpha}}$ continues to be a 1-manifold without boundary. Note that owing to the manifold nature of $\mathcal{M}$, $\{\mathcal{B}(\varpi_{i})\setminus\{\vartheta_{i}\}\}$ never contains a $(1)$-dimensional face, say $\vartheta_{j}$ (where $j<i$), along which $\underline{\widehat{\alpha}}$ was previously expanded. Therefore, because $\mathcal{M}$ is a manifold, the two encounters of $\vartheta_{j}^{1}$ can happen in two different ways, namely: \\
\tbf{Case 1:} While constructing $\underline{\widehat{\alpha}}$ through expansions, any face $\vartheta_{j}^{1}$ can be encountered at most twice - once when it is included in $\underline{\widehat{\alpha}}$ as part of some $\{\mathcal{B}(\varpi_{k})\}$ ($k<j$) and a second time if and when we expand along $\vartheta_{j}^{1}$. Even as we expand along $\vartheta_{j}^{1}$, the two vertices of $\vartheta_{j}^{1}$ stay connected. \\
\tbf{Case 2:} The other possibility of two encounters for the face $\vartheta_{j}^{1}$ is when it is included in $\underline{\widehat{\alpha}}$ as part of some $\{\mathcal{B}(\varpi_{k})\}$ ($k<j$)  and some $\{\mathcal{B}(\varpi_{h})\}$ ($h<j$). In this case, we never expand along $\vartheta_{j}^{1}$. \\
If $\mathcal{M}$ has boundary then our start cell is a coboundary face and upon first expansion, the frame is a manifold with a boundary. Applying the reasoning above, for a given connected component of $\mathcal{G}_s(\mathcal{M})$, the frame of expansion restricted to a single connected component of $\mathcal{G}_s(\mathcal{M})$ is a connected 1-manifold without boundary. To arrive at the more general conclusion that the frames of expansion of all connected components of $\mathcal{G}_s(\mathcal{M})$, pieced together form a single connected 1-complex connecting all 0-cells of manifold $\mathcal{M}$, we have the lemma below:

\begin{lemma} \label{lem:singleframe} Given any sequence of elementary expansions, the frame of a critical cell $\alpha^{2}$ of a manifold $\mathcal{M}$ is always a connected set. Following the final expansion, the frame consists of a set of edges that connects all vertices of the complex.
\end{lemma}
\begin{proof}
Consider without loss of generality, that $\mathcal{M}$ is a manifold without boundary. Then $\mathcal{G}_s(\mathcal{M})$ has a single connected component. Every vertex within the frame that was previously connected, stays connected  by \autoref{lem:staycon}. Since, for a manifold without boundary, the frame always has a single connnected component at every stage of expansion, and since by \autoref{lem:allcon}, all vertices become part of the frame, we arrive at the conclusion that all vertices of the frame form a single connected component at the conclusion of all expansions.\\
The other case, when $\mathcal{G}_s(\mathcal{M})$ has several connected components, we first observe that the frames of each of the connected component stays connected by the same logic as in case of expansions of manifolds without boundary. Also, we observe that in such cases, every connected component will have a 2-cell which is connected to another connected component via a common 0-cell. In fact, if there exist vertices $v_a$ and $v_b$ in two different connected components $C_a$ and $C_b$. $C_a$ and $C_b$ may be interpreted as vertices in the hypergraph then we can first determine a path between $C_a$ and $C_b$ within the component hypergraph $\mathcal{H}_c(\mathcal{M})$. Now, every vertex $C_i$ in the path is a connected component and every hyperedge is a shared 0-cell $v_i$. If the path is written as $C_0,v_1,C_1,\dots,v_i,C_i,\dots v_n,C_n$ where $C_0=C_a$ and $C_n=C_b$. Then for every $C_i$ $1<i<n$, we can determine an internal path (part of the expansion frame) in graph $\mathcal{G}_s(\mathcal{M})$ between $v_{i}$ and $v_{i+1}$. Finally, in graph $\mathcal{G}_s(\mathcal{M})$, we can find a path between $v_a$ and $v_1$ within component $C_0$ and a path between $v_n$ and $v_b$ within component $C_n$ as parts of expansion frames within those components. If we piece together each of the paths from expansion frames of various components of $\mathcal{G}_s(\mathcal{M})$ along the path in the hypergraph $\mathcal{H}_c(\mathcal{M})$, we get a path connecting any two vertices $v_a$ and $v_b$ such that every edge in the path is part of the expansion frame. From this we conclude that all vertices in the complex are connected to each other through edges that lie entirely in the expansion frame. In other words, the expansion frame is a single connected component that connects all vertices of the complex. 
\end{proof}

\begin{lemma} \label{lem:c2b2} Applying the frame based algorithm on a 2-manifold gives us: \[c_2 = \beta_2\] 
\end{lemma}
\begin{proof}\textbf{Case1:} $\beta_2=1$ Suppose the 2-manifold does not have a boundary. Then clearly $\beta_2=1$. Now we will prove that in this case, $c_2$ also equals $1$. Recall that $\mathcal{G}_s(\mathcal{M})$ takes the structure of a simple connected graph and the procedure frameflow is equivalent to a breadth first traversal that begins with a start cell $\varpi$, where $\varpi$ is not included in any of the gradient pairs. However, subsequently every neighboring 2-cell is paired with a 1-cell and added to a queue. The neighbors of the dequeued cell are then scanned and if unpaired, they are paired with the connecting 1-cell as before. This process is continued till all 2-cells are exhausted (which happens at the conclusion of the breadth first traversal). Hence all 2-cells except the start 2-cell $varpi$ form a gradient pair with some 1-cell, giving us $c_2=1$. \\
\textbf{Case2:} $\beta_2=0$ Now, consider the case when the 2-manifold has a boundary. So, we have $\beta_2=0$ and we will prove that $c_2$ also equals $0$. Note that, in this case, $\mathcal{G}_s(\mathcal{M})$ has one or more connected components s.t. each of the connected components has at least one coboundary face. For every component a coboundary face is selected as a start cell and paired with a boundary face to give a gradient pair. Subsequently, as before neighboring 2-cells are paired with connecting 1-cells if they haven't been paired before. Newly paired 2-cells are queued and this process continues till all 2-cells of the connected component are exhausted. In other words, every 2-cell of every connected component is part of a gradient pair giving us $c_2=0$. Hence proved.                 
\end{proof}

\begin{shnote}
Let $\mathcal{B}^1$ be the coboundary of residual complex $\mathcal{M}$. $\mathcal{B}^1[i]$ is part of the coboundary $\mathcal{B}^1$ that intersects with ear $\mathcal{E}_i$. i.e. $\mathcal{B}^1[i]=\mathcal{B}^1\cap\mathcal{E}_i$.
\end{shnote}

\begin{lemma} \label{lem:c0b0} If the complex is made up of a single connected component, then the frame based algorithm gives us $c_0=1$\begin{footnote} {The case when the complex is made up of several connected components can easily be dealt with by applying the algorithm independently to each of the components. In that case $c_0=\mathfrak{C}$ where $\mathfrak{C}$ is the number of connected components}\end{footnote}.
\end{lemma}
\begin{proof} From \autoref{lem:singleframe}, we know that the frame of expansion consists of a single connected component that connects all 0-cells in the manifold. This frame is divided into $\mathscr{N}$ several ears say $\mathcal{E}_i$. Every ear is a 1-dimensional manifold. 
Suppose that we have an open ear then we have 1-dimensional coboundary face in such a ear which we pair with a 0-dimensional boundary face. Subsequently, we follow a path which matches the incident unpaired 0-cell to a neighboring 1-cell and we keep doing this until all 1-cells of the ear are exhausted.   
Now suppose that we have a closed ear. Then we remove one of the 1-cells from the ear (i.e. make it critical). This disconnects the ear into two connected components. We treat these two components of the ears as separate and proceed as in case of open ears.
We now make an inductive argument to prove that the first ear leaves a critical 0-cell. Subsequent addition of ears do not add any criticalities. To see this consider the base case in which we design the flow for the first ear. Here, the flow stops when all 1-cells are exhausted. In this case, for the final 1-cell $\mu$, there is one 0-cell which gets paired with $\mu$ and another incident 0-cells which remains unpaired. It is this 0-cell that becomes the sole critical 0-cell. For induction consider the inductive hypothesis that k-ears have been attached and the number of critical cells remains 1. Now suppose that the (k+1)th ear is attached. If the (k+1)th ear is open then the flow stops with a 1-cell on which one of the incident 0-cells $v_i$  belongs to a ear $\mathcal{E}_i$ where $i<(k+1)$. Either $v_i$ is the sole critical 0-cell or it is paired to another 1-cell belonging to $\mathcal{E}_i$ (by inductive hypothesis). Now, suppose that the (k+1)th ear is closed. Then having detached a 1-cell (which is made critical), we have two disconnected components. For each of the connected components, the flow emanating from subsequent pairing of 0-cells to 1-cells stops when a 1-cell is incident on a 0-cell $v_j$ belonging to a ear $\mathcal{E}_j$ where $j<(k+1)$. Once again by inductive hypothesis either $v_j$ is the sole critical 0-cell or it is paired to another 1-cell belonging to $\mathcal{E}_j$. From this we conclude that $c_0=1$ on attachment of all ears.  
\end{proof}

\begin{theorem}\label{thm:equality} For the frame-based vector field design algorithm, each Morse number equals the Betti number. i.e.
\[c_i = \beta_i\] 
\end{theorem}
\begin{proof} From \autoref{lem:c2b2} and \autoref{lem:c0b0}, we have $c_2=\beta_2$ and $c_0=\beta_0$ respectively. Now, using \autoref{eq:weakmorse} in \autoref{thm:weakmorse}, we have $c_1=\beta_1$. Thus we have $c_i=\beta_i$ for all $i$.
\end{proof}

\subsection{Discussion on Complexity}

Finding coboundary of $\mathcal{M}^{2}$ can be found in linear time
by going through all 2-cells in $\mathcal{M}^{2}$. Finding coboundary
of ears of $\mathscr{M}$ can be found in constant time by mainting
a proper data structure. The ear decomposition of residual complex
$\mathscr{M}$ (which has the structure of a graph) itself takes linear
time. 

Adding a gradient pair to a vector field takes constant time. The
queueing, dequeueing and deletion operations also can be done in constant
time by maintaining appropriate data structures. 

The only nontrivial procedure in the algorithm is \textsf{frameflow}().
Now the\textsf{ frameflow}() procedure can be construed as breadth
first traversal on a semi-graph. We apply this procedure once on $\mathcal{M}^{2}$
and once on each of the ears of $\mathscr{M}$. When traversals from
all ears are counted, we observe that every edge of $\mathscr{M}$
is encountered only once and every vertex $v_{i}$ is encountered
$\mathcal{D}(v_{i})$ number of times where $\mathcal{D}(\cdot)$
indicates degree of a vertex. So, if we sum over all vertices and
edges, the total complexity of \textsf{frameflow}() when applied over
$\mathscr{M}$ is linear in the number of edges of $\mathscr{M}$.
Hence, we see that the design of optimal discrete gradient vector
field using expansion frames takes linear time.

\section{Pseudo-optimality of Random Morse functions \label{sec:Pseudo-optimality-of-Random}}

In this section, we establish the surprising potency of critical cell
cancellations in case of 2-manifolds by using frames. 

\begin{definition}[Pseudo-optimal Vector Field] We define a DGVF to be pseudo-optimal if the optimal DGVF can be obtained from it merely via critical cell cancellations.   
\end{definition}

\begin{definition}[Stable, Unstable Manifolds] The stable manifold of a critical cell $\alpha^{q}$ are all the non-critical cells of dimension $q$ and $q+1$ with gradient paths ending at $\alpha^{q}$. The unstable manifold of a critical cell $\alpha^{q}$ are all the non-critical cells of dimension $q$ and $q-1$ with gradient paths starting at $\alpha^{q}$ and ending at that particular non-critical cell.
\end{definition}

\begin{algorithm}
\begin{algorithmic}[1]      

\Procedure{kingRev}{$\mathfrak{K}^{q},\mathcal{M},\mathcal{C},\mathscr{V},q$}	
	\Repeat \label{lst:line:revgradst}
		\LState{Suppose critical cells $\sigma^{q}$ and $\mathfrak{K}^{q}$ have gradient paths to/from saddle $\gamma^{1}$.} \label{lst:line:revgradst}
		\LState{Subroutine \tbf{sharedSaddle()} finds such a pair $\{\gamma,\sigma\}$ for given $\mathfrak{K}$.}
		\LState{If $\gamma\neq\tsf{NIL}$, then cancel critical pair $(\gamma,\sigma)$} \label{lst:line:revgrad}
	\Until{$(\gamma \neq \tsf{NIL})$} \label{lst:line:revgrad}
		\LState{If $q=2$ AND $\mathfrak{K}^{q}$ has a unique path to $\phi^{1}$, then cancel critical pair $(\phi,\mathfrak{K})$ } \label{lst:line:revgradbd}
\EndProcedure

\Statex

\Procedure{fixBdry}{$d,p,\mathcal{M},\mathcal{C}^{p},\mathcal{B}^d$}
	\ForAll{$1\leq i\leq |\mathcal{B}^d| $} \label{lst:line:looplonelybegin}
		\LState{Let $B_i\colonequals \mathcal{B}^d[i]$ and let $ b_{ij}$ be a boundary face of $B_i$.}
		\If{$\langle b_{ij},B_i \rangle \notin \mathcal{V} \tsf{ AND } b_{ij} \tsf{ is critical}$}
			\LState{Let $\langle \theta_i,B_i\rangle$ be a gradient pair}
			\If{$\theta_i\neq\tsf{NIL}$ and $\theta_i$ is not a boundary face of $B_i$}	
				\LState{Find a gradient path from some critical cell $\alpha^{d}$ to $\theta_{i}$ and reverse it.}
			\EndIf
			\LState{Add gradient pair $\langle b_{ij},B_i\rangle $ to vector field $\mathscr{V}$}
		\EndIf
	\EndFor \label{lst:line:looplonelyend}
\EndProcedure

\Statex

\Procedure{findKing}{$d,p,\mathcal{M},\mathcal{C}^{p},b^q,i$}
	\LState {\algorithmicif $\left(p=2 \tsf{ AND } \mathcal{C}^{p}\neq \tsf{NIL}\right) \tsf{ OR } (p=1 \tsf{ AND } i=1)$ \algorithmicthen\   $\tbf{remFrom}(\mathcal{C}^p,\mathfrak{K})$;} \label{lst:line:kingstart}
	\LState {\algorithmicelse\ \algorithmicif\ $p=1$ \algorithmicthen\ $\mathfrak{K}\colonequals b^q[1]$;\quad $\tbf{remFrom}(b^q,\mathfrak{K})$;}
	\LState {\algorithmicelse\ $\mathfrak{K}\colonequals\tsf{NIL}$;}
	\LState \algorithmicend\ \algorithmicif \label{lst:line:kingend}
	\LState{$\mathfrak{K}\colonequals\tbf{selectRandomly}(\mathcal{C}^{p})$; \;\tbf{return }$\mathfrak{K}$;}
\EndProcedure

\Statex

\Procedure{processComplex}{$\mathcal{M},i,d,p$}
	\LState{$\mathcal{C}[1:d]\leftarrow\tbf{identifyCritical}(\mathcal{M},\mathscr{V})$}
	\LState{$\tbf{findBdry}()$ finds $\mathcal{B}^d$ \& $b^{d-1}$ the cobdry. and bdry. of $\mathcal{M}$ resp.}
	\LState{$\tbf{fixBdry}(d,p,\mathcal{M},\mathcal{C}^{d},\mathcal{B}^d)$}
	\While{($\mathfrak{K}=\tbf{findKing}(d,p,\mathcal{M},\mathcal{C}^{p},b^{d-1},i)\neq \tsf{NIL}$} \label{lst:line:callfirstking}
		\LState{$\tbf{kingRev}(\mathfrak{K},\mathcal{M},\mathcal{C},\mathscr{V},p)$}
	\EndWhile \label{lst:line:callfirstkingend}
\EndProcedure

\Statex

\Procedure{kingFlow}{$\mathfrak{M},d,\mathscr{V}$}	
	\LState{Divide  $\mathfrak{M}$ into manifolds $\mathcal{M}_{1},\mathcal{M}_{2},\dots,\mathcal{M}_{K}$ s.t. $\mathcal{M}_{i}\cap\mathcal{M}_j$ is 0-dimensional.}
\LState {\algorithmicfor\ $1\leq i\leq K $ \algorithmicdo\ \quad $\tbf{processComplex}(\mathcal{M}_i,i,2,2)$\quad \algorithmicend\ \algorithmicfor} 							\label{lst:line:compfor}
	
 \label{lst:line:compforend}
\LState{$\mathcal{E}[1:\tsf{numEars}]\colonequals\tbf{earDecompose}(\mathscr{M}$)}
\LState {\algorithmicfor\ $1 \leq  i \leq \tsf{numEars}$ \algorithmicdo\ \quad $\tbf{processComplex}(\mathcal{E}[i],i,1,0)$ \quad \algorithmicend\ \algorithmicfor}
\EndProcedure

\end{algorithmic}

\protect\caption{Optimal DGVF Redesign Using Critical Cell Cancellations}

\label{alg:optcancel}
\end{algorithm}

\begin{shnote}
Given a connected pseudomanifold complex $\mathfrak{M}$, divide  $\mathfrak{M}$ into several connected components $\mathcal{M}_{1},\mathcal{M}_{2},\dots,\mathcal{M}_{K}$ s.t. $\mathcal{M}_{i}\cap\mathcal{M}_j$ is 0-dimensional. i.e. any of the two manifolds (with boundary) $\mathcal{M}_{i},\mathcal{M}_{j}$ may intersect only along points (but not along edges). If $\mathfrak{M}$ is a manifold without boundary then $\mathfrak{M}$ will have only one connected component. $\mathcal{M}_i$ are essentially the connected components of the semigraph $\mathcal{G}_s(\mathcal{M})$ defined in \autoref{note:semigraph}.
\end{shnote}

\begin{lemma} \label{lem:revc2nobd} If $\mathfrak{M}$ is a manifold without boundary then after invoking the procedure \textbf{kingRev()}, we obtain a connected expansion frame. Moreover, $c_2=\beta_2=1$.  
\end{lemma}
\begin{proof}
Suppose a vector field $\mathscr{V}$ on manifold $\mathfrak{M}$ (without boundary) has a single critical cell. Then from \autoref{lem:singleframe}, we get a single connected expansion frame connecting all vertices of $\mathfrak{M}$. Instead, if $\mathfrak{M}$ is a manifold without boundary and if we have more than one critical 1-cells, then consider the unstable manifold of some chosen critical cell $\mathfrak{K}$. Since the unstable manifold of $\mathfrak{K}$ doesnot include the entire manifold $\mathfrak{M}$, the stable manifold has a 1-dimensional manifold as its boundary. From \cite{Lee00}, we know that, if $\mathcal{M}$ is an n-dimensional manifold with boundary, then the boundary  of $\mathcal{M}$ is an (n-1)-dimensional manifold (without boundary) when endowed  with the subspace topology. Therefore, the boundary of the unstable manifold is a 1-dimensional manifold  without boundary (i.e.  it consists of one or more disjoint circles). Clearly the 1-cells belonging to this boundary are not part of the 2-flow  of $\mathscr{V}$, else they wouldn't be part of the boundary of the unstable  manifold of  $\mathfrak{K}$.  So, the 1-cells belonging to this boundary are either part of the 1-flow of $\mathscr{V}$ or they are critical. Consider one of the disjoint circles that forms part of the boundary of the unstable manifold. If all the cells on this circle are part of the 1-flow then it will form a cycle. Hence there exists at least one critical 1-cell on the boundary of the unstable manifold. Let $\gamma$ be a critical that lies on the boundary of the unstable manifold of $\mathfrak{K}$. Clearly, there exists only one gradient path from $\mathfrak{K}$ to $\gamma$. $\gamma$ is also incident on a 2-cell say $\sigma_1$ that does not lie in the unstable manifold of $\mathfrak{K}$. Suppose $\sigma_1$ is itself a critical 2-cell, then $\gamma$ lies on the boundary of unstable manifolds of the two critical 2-cells $\mathfrak{K}$ and $\sigma_1$. Otherwise suppose that $\sigma_1$ is matched. Because the simplicial complex $\mathfrak{M}$ is a manifold, it is possible to trace any inverted gradient path on $\mathfrak{M}$ (such a unique inverse gradient path exists). Therefore, we trace the inverted gradient path $\gamma,\sigma_1,\dots$ until we reach a critical 2-cell (say $\sigma_k$) from which this path emanates. In any case, we can find a critical 1-cell $\gamma$ which is shared by critical cells $\mathfrak{K}$ and some other critical 2-cell say $\sigma$. In this case, because gradient path from $\sigma$ to $\gamma$ is unique we can invert this gradient path as shown in Line \autoref{lst:line:revgrad} of Procedure \textbf{kingRev()} of \autoref{alg:optcancel}. Once this cancellation is done, the unstable manifold of $\sigma$ becomes part of the new unstable manifold of $\mathfrak{K}$. Once again we search a critical 1-cell $\gamma_2$ on the boundary of the unstable manifold s.t. which also lies on the boundary of unstable manifold of some other critical 2-cell (distinct from $\mathfrak{K}$). If such a pair of critical cells is found then we cancel it and this procedure is repeated until all critical 2-cells belong to the unstable manifold of $\mathfrak{K}$ (or alternatively all critical 1-cells have two gradient paths from $\mathfrak{K}$.) Basically  this means that $\mathfrak{M}$ is a manifold without boundary that has a unique critical 2-cell. i.e. $c_2=1$. Since, $\mathfrak{M}$ is a 2-manifold without boundary, $\beta_2=1$. Finally, from Case 1 of \autoref{lem:c2b2}, we arrive at the conclusion that the expansion frame is a connected 1-manifold that includes all 0-cells of $\mathfrak{M}$.
\end{proof} 

If $\mathfrak{M}$ is a manifold without boundary then $\mathcal{G}_s(\mathcal{M})$ has a single connected component and the for loop described in Lines \ref{lst:line:compfor}-\ref{lst:line:compforend} of Procedure \textbf{kingFlow()} in \autoref{alg:optcancel} gets executed only once. Also the while loop described in Lines \ref{lst:line:callfirstking}-\ref{lst:line:callfirstkingend} of Procedure \textbf{processComplex()} in \autoref{alg:optcancel} gets executed only once for manifolds without boundary. This is because for any critical 2-cell $\mathfrak{K}$, you always find another critical 2-cell $\sigma$ s.t. both $\mathfrak{K}$ and $\sigma$ have a gradient path to a common 1-cell $\gamma$ unless the unstable manifold of $\mathfrak{K}$ covers the entire manifold $\mathfrak{M}$.\\
The situation is however much different for a manifold with boundary. For such a manifold the for loop and the while loop may run several iterations.

\begin{lemma} \label{lem:revc2bd} If $\mathfrak{M}$ is a manifold with boundary then after invoking the procedure \textbf{kingRev()}, we obtain a connected expansion frame. Moreover, $c_2=\beta_2=0$.  
\end{lemma}
\begin{proof} We will examine the effect of the algorithm on one of the connected components $\mathcal{M}_i$ of $\mathcal{G}_s(\mathcal{M})$. Consider the unstable manifold of a critical 2-cell $\mathfrak{K}$. From \cite{Lee00}, we know that, if $\mathcal{M}$ is an n-dimensional manifold with boundary, then the boundary  of $\mathcal{M}$ is an (n-1)-dimensional manifold (without boundary) when endowed  with the subspace topology. Hence, the boundary of this unstable manifold will be a 1-manifold without boundary (i.e. a disjoint set of circles). The 1-cells on any one of these circle are involved only in 1-flows or they are critical. But all, cells of a circle can not be involved in 1-flow as this would lead to a cycle in the vector field. So, every circle must contain a 1-cell, say $\gamma$ that is critical. $\mathfrak{K}$ has only one gradient path to $\gamma$. There exists a second gradient path that ends at $\gamma$. This gradient path either emanates from another critical 2-cell say $\sigma$ or it emanates from a boundary face. Assume the case where a path to $\gamma$ emanates from $\sigma$. In this case, the pair $(\gamma,\sigma)$ is detected and cancelled in Lines \ref{lst:line:revgradst}-\ref{lst:line:revgrad} of Procedure \textbf{kingRev()} in \autoref{alg:optcancel}. In fact, every such pair  $(\gamma,\sigma)$ for a given $\mathfrak{K}$ is detected and cancelled in the loop Lines \ref{lst:line:revgradst}-\ref{lst:line:revgrad} of Procedure \textbf{kingRev()} in \autoref{alg:optcancel}. So finally every critical 1-cell say $\phi$ in the boundary of the unstable manifold of $\mathfrak{K}$ will have a second path emanating from a boundary face. In this case, the pair of critical cells $(\phi,\mathfrak{K})$ is detected and cancelled as shown in Line \ref{lst:line:revgradbd} of Procedure \textbf{kingRev()} in \autoref{alg:optcancel}. Suppose that $\mathcal{M}_i$ continues to have critical cells that are not cancelled, then a new critical king cell $\mathfrak{K}$ is selected and the same procedure as described above is repeated in a loop shown in Lines \ref{lst:line:callfirstking}-\ref{lst:line:callfirstkingend} of Procedure \textbf{processComplex} in \autoref{alg:optcancel}. We exit from the loop provided there are no other critical 2-cells to process in the list $\mathcal{C}^p$. In case of manifolds with boundary every critical 2-cell processed as a king cell $\mathfrak{K}$ is itself cancelled along with cancelling all the neighboring critical 2-cells that share gradient paths to the same saddles as $\mathfrak{K}$. Having processed $\mathcal{M}_i$ in this manner, we are assured that eventually $\mathcal{M}_i$ has no critical 2-cells. In fact every 2-flow for $\mathcal{M}_i$ emanates strictly from boundary faces. Using an argument similar to that in Case 2 of \autoref{lem:c2b2}, we know that the frame of expansion of a boundary face is a connected set. Consider the first such boundary face $b_1$, with a frame of expansion which is a connected 1-manifold. Every 1-cell belonging to the frame of expansion of $b_1$ has a second gradient path emanating from other boundary cells $\{b_i\}$. Since the $\mathcal{M}_i$ is a manifold with boundary, given any pair of boundary faces $b_i,b_j$, we can find a type 2 connected 2-path between them. Consider all the 1-cells in some such type 2-connected path between $b_i$ and $b_j$. Every 1-cell either lies in the frame of expansion of two boundary faces or is involved in 2-flow with a regular 2-cell. This gives us a sequence of frames of expansion of boundary faces $b_{k_1},b_{k_2},\dots,b_{k_L}$ that are sequentially pairwise connected and s.t. $b_{k_1}=b_i$ and $b_{k_L}=b_j$. Since this procedure can be applied to any two boundary faces (with expansion frames), we conclude that the set of frames of expansion of all boundary faces is a connected set, which we refer to as the expansion frame of  $\mathcal{M}_i$. To see that the frame of expansions of all $\mathcal{M}_i$ form a single connected set, we consider the component hypergraph $\mathcal{H}_{c}(\mathcal{M})$. We then use the same line of reasoning as used in \autoref{lem:singleframe}, to conclude that the expansion frame of $\mathfrak{M}$ is connected. 
Also, following all critical cell cancellations since there are no more critical 2-cells for $\mathcal{M}_i$, we have $c_2=\beta_2=0$ for each $\mathcal{M}_i$. So, we have also have $c_2=\beta_2=0$ for $\mathfrak{M}$        
\end{proof}

\begin{shnote}\label{note:onecritical} The residual 1-complex $\mathscr{M}$ is essentially the expansion frame of $\mathfrak{M}$ following cancellation of critical cell pairs of dimensions $1,2$. Since there exists a preordained 1-flow (without cycles) on $\mathscr{M}$, clearly given the mechanism of discrete Morse theory, there must exist at least one critical 0-cell. (A sub-optimal 1-flow may have more than one critical 0-cells. But at least one is guaranteed.) The first ear is chose to be one that includes at least one of these critical 0-cells. Also the ear decomposition follows a special procedure. The number of ears are determined by the number of unstable manifolds of boundary 0-cells and critical 1-cells. The first ear is either an unstable manifold of a boundary 0-cell or a critical 1-cell that includes at least one critical 0-cell. The second ear is a 1-manifold that is incident on a 0-cell that belongs to the $1^{st}$ ear and includes all the 0-cells and 1-cells of an unstable manifold of a boundary 0-cell or a critical 1-cell that aren't already included in the first ear. The $k^{th}$ ear is a 1-manifold that is incident on a 0-cell that belongs to one of the previous $(k-1)$ ears and includes all 0-cells and 1-cells of an unstable manifold of a boundary 0-cell or a critical 1-cell that aren't already included as part of the previous $(k-1)$ ears. Every ear (apart from the first ear), has at least one 0-cell in its 0-dim. boundary whereas every ear may have at most two 0-cell in its 0-dim. boundary. The first boundary cell of the ear $b^{0}[1]$ is incident on one of the previous ears.  The second boundary cell $b^{0}[2]$ may or may not be incident on any of the previous ears. 
\end{shnote}

\begin{lemma} \label{lem:revc0} On applying a series of critical cell cancellations, the connected expansion frame has $c_0=\beta_0=0$   
\end{lemma}
\begin{proof} We shall make an inductive argument. The idea is that the first ear will have a 0-cell that is critical. Subsequent ears attached to the first ear have no 0-dimensional critical cells. Note that all ears are 1-dimensional manifolds (topological circles or topological line segments) \\
\textbf{Base Case:} Suppose that we start with the first ear. Suppose that the first ear is a closed loop (i.e. a topological circle). From \autoref{note:onecritical} our first ear has at least one critical 0-cell. Suppose we call it $\varrho^0$. In this case $\varrho$ becomes our king critical cell $\mathfrak{K}$.  If there exist two gradient paths to $\varrho$ from a saddle, then clearly we do not have any criterion for cancellation. Instead if we have a single gradient path from the saddle $\gamma^{1}$ and suppose there exists another gradient path from the $\gamma$ to some other minima $\varphi^{0}$ then from Lines \ref{lst:line:revgradst}-\ref{lst:line:revgrad} of Procedure \textbf{kingRev()} in \autoref{alg:optcancel}, we cancel critical pair $(\varphi,\gamma)$ and as a result have a single critical 0-cell in the first ear. The last possibility the first ear consists of the unstable manifold of a critical 0-cell $\varsigma$ that $\varrho^0$   \\
\textbf{Induction step:} By the inductive hypothesis, we have processed $(k-1)$ ears so far and for all the $(k-1)$ ears taken together, we have only one critical 0-cell (namely the one that was encountered in the very first ear.) Now, we need to establish that on attachment of the $k^{th}$ ear we do not introduce any new critical 0-cells. Note that for $k^{th}$ ear we start with $b^{0}[1]$ as the king cell $\mathfrak{K}$, where $b^{0}[1]$ is incident on one of the previous ears (i.e. it may either be our original critical cell $\varrho^{0}$, or it may be some regular 0-cell from one of the earlier $(k-1)$ ears. Like all other ears, the $k^{th}$ ear is an unstable manifold of a boundary 0-cell or a critical 1-cell. If it is the unstable manifold of a boundary 0-cell $\varsigma$ then we do not have anything to prove as the flow for this cell will simply start with $\varsigma$ and end at $b^{0}[1]$ without introducing any criticalities. If the $k^{th}$ ear is topologically a loop, then  $b^{0}[1]$ has two gradient paths from some saddle $\gamma$ and hence the criterion for cancellation is not satisfied. Yet another case is when $b^{0}[1]$ and $b^{0}[2]$ are both incident on one of the earlier $k$ ears. In this case, $b^{0}[1]$ is either $\varrho$ or a regular 0-cell and $b^{0}[2]$ is certainly a regular cell. Also, there does not exist any other critical 0-cell in this ear because the $k^{th}$ ear, in this case, is an unstable manifold of a saddle.    The only interesting case is when $k^{th}$ ear is topologically a line segment s.t. $b^{0}[2]$ is critical and the saddle $\gamma$ has one gradient paths each to $b^{0}[1]$ and $b^{0}[2]$. Since, for $k^{th}$ ear we start with $b^{0}[1]$ as king cell $\mathfrak{K}$, we end up cancelling $\gamma$ and $b^{0}[2]$, making the $k^{th}$ ear an unstable manifold of boundary cell $b^{0}[2]$. In each of the cases, we ensure that either the $k^{th}$ ear did not have any critical 0-cell to begin with or if there does exist a critical 0-cell, then it is cancelled. Hence proved.      
\end{proof}

\begin{theorem}Every discrete gradient vector field on a 2-manifold is pseudo-optimal.\end{theorem} 
\begin{proof} 
Suppose that at the end of the first call to Procedure \textbf{findKing()} from Line \ref{lst:line:callfirstking} of Procedure \textbf{processComplex()} in \autoref{alg:optcancel}, $\mathfrak{K}^{2}$ is not \textsf{NIL}. Then, we claim that the unstable manifold of $\mathfrak{K}^{2}$ does not have any critical 1-cells that are boundary faces. This is because, if $\mathfrak{K}^{2}$ did have any boundary critical 1-cells in its unstable manifold, it would have got cancelled in the loop shown in Lines \ref{lst:line:looplonelybegin}-\ref{lst:line:looplonelyend} in Procedure \textbf{fixBdry()} in \autoref{alg:optcancel}. In fact, more generally every critical 2-cell at the end of first call to  Procedure \textbf{findKing()} will have no critical boundary 1-cells in their respective unstable manifolds. If $\mathfrak{M}$ is a manifold without boundary, then by \autoref{lem:revc2nobd}, we have $c_2=\beta_2=1$ and we get a connected frame of expansion in form of residual complex $\mathscr{M}$. Instead, if $\mathfrak{M}$ is a manifold with boundary, then from \autoref{lem:revc2bd} we obtain $c_2=\beta_2=0$ and a connected frame of expansion in form of residual complex $\mathscr{M}$. Given a connected frame of expansion $\mathscr{M}$, \label{lem:revc0} guarantees that we have $c_0=\beta_0=1$. Finally, using Weak Morse Inequlity we obtain $c_1=\beta_1$. Hence, we prove that for a 2-manifold, given an arbitrary vector field $\mathscr{V}_1$ merely by using critical cell cancellations, we may obtain the optimal vector field $\mathscr{V}_2$. In other words, every gradient vector field on a 2-manifold is pseudo-optimal. 
\end{proof}

\section{The Topological Explanation for simplicity of computation of $\ensuremath{H(\mathcal{M}^{2},\mathbb{A})}$}

We compute homology using \autoref{alg:hom}. Also we assume Weak
Morse Optimality Condition as defined in \autoref{def:WMOC} on the
input.

As we can see from arguments in \autoref{sec:Frames-of-Expansion},
the topological explanation for simplicity of computation of homology
groups for 2-manifolds is:
\begin{enumerate}
\item On 2-manifolds optimal Morse functions are perfect. In fact, 2-manifolds
admit readily computable perfect Morse functions.
\item A 2-manifold has optimal $c_{2}=0$ or $c_{2}=1$, which can be figured
out in linear time by examining whether or not it has a boundary.
\item We define and apply \emph{frames of expansion} an elementary homotopy
theory construct to design our algorithm.
\item It can be seen that irrespective of what traversal method we use to
traverse the graph like connectivity structure of 1-cells and 2-cells
of a 2-manifold, the frame of expansion remains connected. Furthermore,
this connectivity guarantees that optimal $c_{0}=1$. 
\item Finally weak Morse Inequality guarantees that our $c_{1}$ is optimal.
i.e. $c_{1}=\beta_{1}$.
\item Moreover, our dynamic programming based boundary operator computation
algorithm is pseudo-linear time (which becomes strictly linear assuming
WMOC).
\item Finally, assuming WMOC, the application of Smith Normal Form (a supercubical
time algorithm) on input of constant size is inexpensive. 
\item The pseudo-optimality of arbitrary discrete Morse functions as outlined
in \autoref{sec:Pseudo-optimality-of-Random} further strengthens
our argument about simplicity of computing optimal discrete Morse
functions. 
\end{enumerate}

\section{Concluding Remarks }

In this work, we provide a nearly linear time algorithm for computing
homology (with arbitrary coefficients) on 2-manifolds - the first
such algorithm. This is particularly useful to compute homology of
2-manifolds that may have torsion elements. The design involves the
introduction and usage of an elementary simple homotopy construct
that we call \emph{expansion frames.} Having designed the optimal
Morse function in linear time, we use a dynamic programming based
pseudo-linear time boundary operator algorithm for computing the Morse
boundary operator. Assuming the sum of Betti numbers is a small constant
compared to the size of the complex, the Smith Normal Form is applied
to a very smal input, giving us near-linearity. Finally, using the
notion of expansion frames, we prove an unexpected result in discrete
Morse theory: Start with an arbitrary DGVF on a 2-manifold and one
may obtain an optimal DGVF merely by application of critical cell
cancellations. 

\bibliographystyle{acm}
\bibliography{morseopt}

\section*{Appendix}

\section{Elementary Algebraic Topology\label{sec:Elementary-Algebraic-Topology}}

\begin{figure}
\shadowbox{\begin{minipage}[t]{1\columnwidth}%
\includegraphics[scale=0.4]{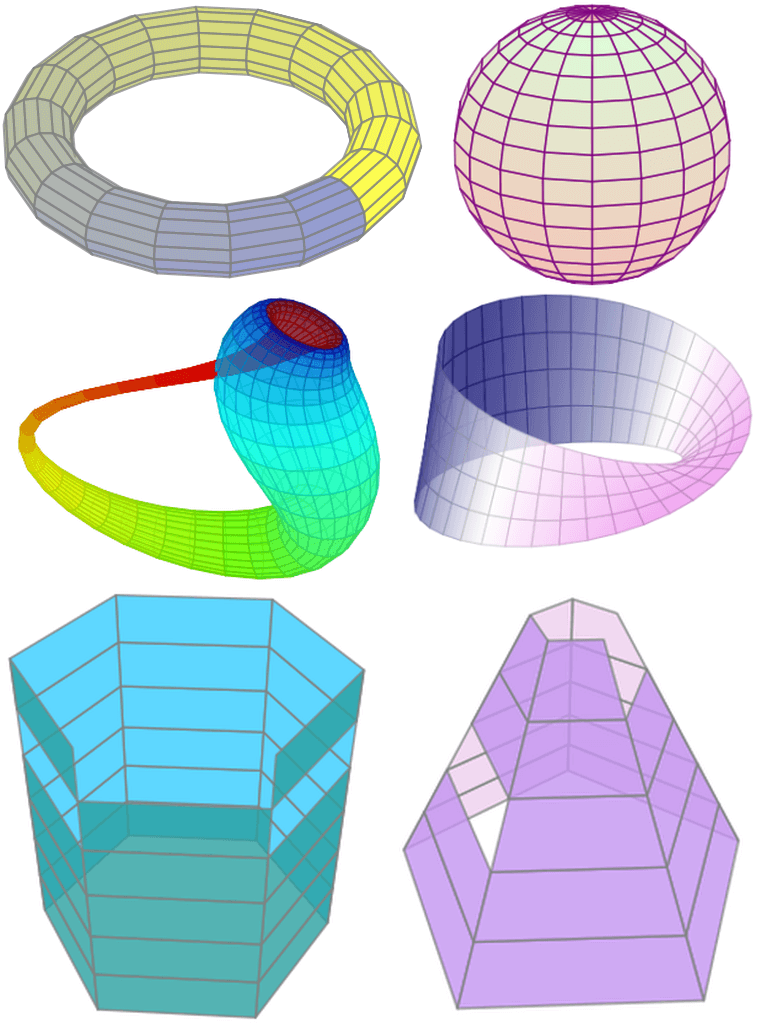}\cite{Le03a}%
\end{minipage}}

\protect\caption{2-manifolds}
\end{figure}

\begin{definition}[Simplicial Complex] A simplicial complex $\mathcal{K}$ is a set of vertices and a collection $\mathcal{L}$ of subsets of vertices called faces. All faces satisfy the following property: The subset of a face is also a face. (i.e. $\mathcal{B}\in\mathcal{L}, \mathcal{A}\in\mathcal{B}\Longrightarrow\mathcal{A}\in \mathcal{L}$). Maximal faces w.r.t. inclusion are known as \textbf{facets}. The \textbf{dimension} of a face $\mathcal{B}$ is defined to be $ |\mathcal{B} |-1$. The dimension of the simplicial complex itself is the maximum over the dimension of its faces.  
\end{definition}  

\begin{definition}[Open Cell] An n-dimensional open cell is a topological space that is homeomorphic to an open ball.  
\end{definition}

\begin{definition} [Cell Complex] A hausdorff topological space $X$ is called a finite cell complex if
\begin{enumerate} 
\item $X$ is a disjoint union of open cells $\{D_i^{n}\}$ where $D_i^{n}$ is an open $n$-cell. ($i \in I$ where $I$ is the indexing set.) 
\item For each open cell $D_i^n$ there is a map $\phi_i^{n}: B^{n} \to X$ such that $\phi_i^{n}$ restricted to the interior of the closed ball $B^{n}$ defines a homeomorphism to $D_i^{n}$ and such that $\phi_i^{n}(S^{n - 1})$ is contained in the $(n - 1)$-skeleton of $X$. (The $k$-skeleton of $X$ is the union of all open cells $D_i$ of dimension $r \leq k$).
\item Finally, a set $\alpha$ is closed in $X$ if and only if $\alpha \cap \overline D_j$ is closed in $\overline{D_j}^{n}$ for each cell $D_j^{n}$. Note that $\overline{D_j}^{n} = \phi_i^n(B^{n})$.  
\end{enumerate} 
A cell complex is said to be regular if each $\phi_i^n$ is a homeomorphism and if it sends $S^{n -1}$ to a union of cells in the $(n - 1)$-skeleton of $X$. 
\end{definition}

In lay man terms, to construct a cell complex you start with points
$\mathcal{D}_{i}^{0}$, then glue on lines $\mathcal{D}_{i}^{1}$
to $\mathcal{D}_{i}^{0}$, then glue discs $\mathcal{D}_{i}^{2}$
to $\mathcal{D}_{i}^{1}$ and $\mathcal{D}_{i}^{0}$ and so on. Therefore
a cell complex is a topological space constructed from a union of
objects called cells, which are balls of some dimension, glued together
on boundaries. Cell complexes are the most convenient object to do
Algebraic Topology. But to simplify the discussion, we will instead
provide a basic presentation of simplicial homology. 

\begin{notation} \textbf{Boundary \& Coboundary of a simplex} $\sigma$:  We define the boundary $bd(\sigma)$ and respectively coboundary $\eth\:\sigma$ of a simplex as\\ $\eth \: \sigma \,= \: \{ \tau \, | \, \tau \prec \sigma \}$ $\delta \: \sigma = \{ \rho \, | \, \sigma \prec \rho \}$ \end{notation}  

Homology groups are the most important and general topological invariants
of simplicial and cubical complexes, that are also computationally
feasible. At the heart of it, Algebraic Topology is essentially the
use of Linear Algebra to compute combinatorial topological invariants
of a give space. Given a simplicial complex $W$, can define simplicial
$q$-chains, which are formal sums of $q$-simplices $\sum_{s\in S}a_{i}s_{i}$
where the $a_{i}$ are integer coefficients. The abelian group of
sums of $k$-simplices under addition is called the Chain Group and
denoted by $C_{q}(W,\,\mathbb{Z})$. The $n$-simplex $\triangle=\{v_{0},v_{1},\cdots,v_{n}\}$with
standard orientation is denoted $+\left[v_{0},v_{1},\cdots,v_{n}\right]$.
Consider the permutation group of $n$-letters on the vertices of
$\triangle$. The set of permutations fall into 2 equivalence classes:
even permutations and odd permutations. The set of even permutations
induce the positive orientation $+\left[v_{0},v_{1},\cdots,v_{n}\right]$
whereas the set of odd permutations induce the negative orientation
$-\left[v_{0},v_{1},\cdots,v_{n}\right]$.

For each integer $q$, $C_{q}(W)$ is the free abelian group generated
by the set of oriented $q$-simplices of $W$. Let $W_{q}$ be the
total number of $q-$dimensional simplices for simplicial complex
$W$. Then, one can show that $C_{q}\cong\mathbb{Z}^{W_{q}}$.

The boundary map $\partial_{q}$ is defined to be the linear transformations
$\partial_{q}\,:\,C_{q}\rightarrow C_{q-1}$.

Examples of such operations are given in \textbf{Fig.E3} and \textbf{Fig.E4.}

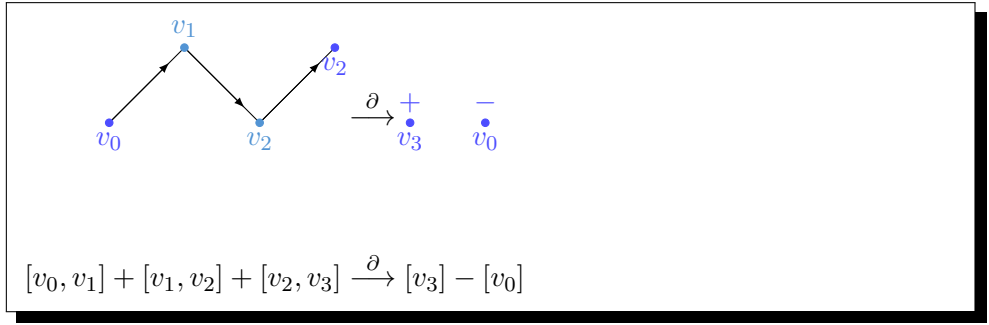
\begin{figure}
\shadowbox{\begin{minipage}[t]{1\columnwidth}%
\begin{tikzpicture}          
\draw[color=black] (0,0) -- (1,1) -- (2,0) -- (3,1); 
\draw[->, >=latex]          (0,0) -- (.8,.8); 
\draw[->, >=latex] (1,1) -- (1.8,.2); 
\draw[->, >=latex]          (2,0) -- (2.8,.8);
\filldraw          [color=blue!70] (0,0) node[below] {$v_0$} circle (.05); \filldraw          [color=hexcolor0x4f91d3!95!hexcolor0x93ccea] (1,1) node[above] {$v_1$} circle (.05); \filldraw          [color=hexcolor0x4f91d3!95!hexcolor0x93ccea] (2,0) node[below] {$v_2$} circle (.05); 
\filldraw          [color=blue!70] (3,1) node[below] {$v_2$} circle (.05); 

\filldraw          [color=blue!70] (4,0) node[below] {$v_3$} node[above] {$+$}          circle (.05); 
\filldraw [color=blue!70] (5,0) node[below] {$v_0$}          node[above] {$-$ } circle (.05); 
\node at (3.5,0.2){$\stackrel{\partial}{\longrightarrow}$};
\node at (2.2,-2)          {$[v_0,v_1]+[v_1,v_2]+[v_2,v_3]\stackrel{\partial}{\longrightarrow}          [v_3]-[v_0]$};         
\end{tikzpicture}   %
\end{minipage}}

\label{fig:homo1}

\protect\caption{Dim 1 Boundary Operator}
\end{figure}

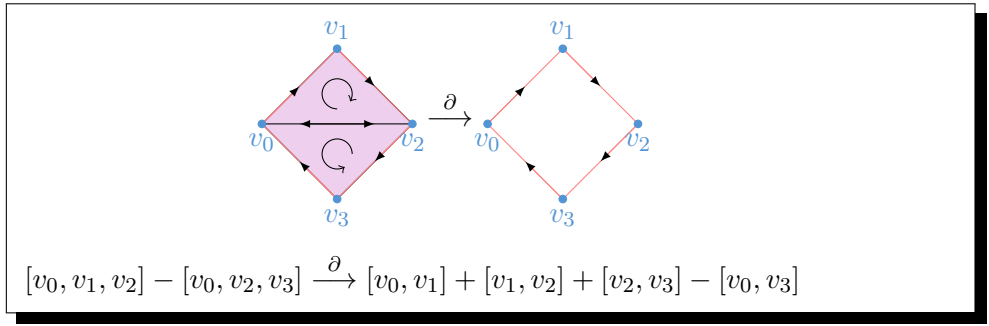
\begin{figure}
\shadowbox{\begin{minipage}[t]{1\columnwidth}%
\begin{tikzpicture}            
\filldraw[color=Plum!50, draw=black] (0,0) -- (1,1) -- (2,0)-- (1,-1)--           cycle; 
\draw[draw=Red!50] (0,0) -- (1,1) -- (2,0)-- (1,-1)-- cycle;
\draw[draw=Red!50] (3,0) -- (4,1) -- (5,0)-- (4,-1)-- cycle;            
\draw[->] (1,.2) arc (270:-30:.2); 
\draw[->] (1.2,-.4) arc (0:300:.2); 
\draw[->, >=latex,draw=Red!50] (0,0) --            (.5,.5);  
\draw[->, >=latex,draw=Red!50] (1,1) -- (1.5,.5); 
\draw[->, >=latex]            (2,0) -- (0.5,0);
\draw[->, >=latex]            (0,0) -- (1.5,0); 
\draw[->, >=latex,draw=Red!50]            (2,0) -- (1.5,-0.5);
\draw[->, >=latex,draw=Red!50]            (1,-1) -- (0.5,-0.5); 
\draw[->, >=latex,draw=Red!50] (3,0) -- (3.5,.5);            
\draw[->, >=latex,draw=Red!50] (4,1) -- (4.5,.5); 
\draw[->, >=latex,draw=Red!50] (5,0) --            (4.5,-0.5); 
\draw[->, >=latex,draw=Red!50] (4,-1) --            (3.5,-0.5); 
\filldraw [color=hexcolor0x4f91d3!95!hexcolor0x93ccea] (0,0) node[below] {$v_0$}            circle (.05); \filldraw [color=hexcolor0x4f91d3!95!hexcolor0x93ccea] (1,1) node[above] {$v_1$}            circle (.05); \filldraw [color=hexcolor0x4f91d3!95!hexcolor0x93ccea] (2,0) node[below] {$v_2$}            circle (.05);
\filldraw [color=hexcolor0x4f91d3!95!hexcolor0x93ccea] (1,-1) node[below] {$v_3$}            circle (.05);

\filldraw [color=hexcolor0x4f91d3!95!hexcolor0x93ccea] (3,0) node[below] {$v_0$}            circle (.05); \filldraw [color=hexcolor0x4f91d3!95!hexcolor0x93ccea] (4,1) node[above] {$v_1$}            circle (.05); \filldraw [color=hexcolor0x4f91d3!95!hexcolor0x93ccea] (5,0) node[below] {$v_2$}            circle (.05);
\filldraw [color=hexcolor0x4f91d3!95!hexcolor0x93ccea] (4,-1) node[below] {$v_3$}            circle (.05);

\node at (2,-2)            {$[v_0,v_1,v_2]-[v_0,v_2,v_3]\stackrel{\partial}{\longrightarrow}            [v_0,v_1]+[v_1,v_2]+[v_2,v_3]-[v_0,v_3]$}; 
\node at (2.5,0.2){$\stackrel{\partial}{\longrightarrow}$};
\end{tikzpicture}%
\end{minipage}}

\label{fig:homo2}\protect\caption{Dim II Boundary Operator}
\end{figure}

\begin{figure}
\shadowbox{\begin{minipage}[t]{1\columnwidth}%
\begin{tikzpicture}       
\fill [color=blue!30] (-1.7,0)--(1,-0.5)--(0,1.7); 
\draw [color=red!80] (-1.7,0)--(1,-0.5)--(0,1.7)--(-1.7,0); 
\node [below] at (-1.7,0) {$v_0$}; 
\node [below] at (1,-0.5) {$v_1$}; 
\node [above] at (0,1.7) {$v_2$};  
\fill [color=red!80] (-1.7,0) circle (2.5pt);
\fill [color=red!80] (1,-0.5) circle (2.5pt); 
\fill [color=red!80] (0,1.7) circle (2.5pt); 

\node [above] at (-0.45,-1.7) {$[ v_0,v_1,v_2]$}; 
\node [above] at (1.7,-1.7) {$\xrightarrow{\partial_2}$}; 

\draw [->][color=red!80] (2,-0.0)->(4.5,-0.5);
\draw [->][color=red!80] (4.5,-0.5)->(3.5,1.5); 
\draw [->][color=red!80] (3.5,1.5)->(2,0); 
\node [below] at (2,0) {$v_0$};  
\node [below] at (4.5,-0.5) {$v_1$};  
\node [above] at (3.5,1.5) {$v_2$};  
\fill [color=red!80] (2,0) circle (2.5pt); 
\fill [color=red!80] (4.5,-0.5) circle (2.5pt);
\fill [color=red!80] (3.5,1.5) circle (2.5pt); 

\node [above] at (3.25,-1.7) {$[ v_0,v_1]+[v_1,v_2]$}; 
\node [above] at (3.25,-2.1) {$+[v_2,v_0]$};
\node [above] at (5,-1.7) {$\xrightarrow{\partial_1}$}; 

\fill [color=red!80] (5.5,0) circle (2.5pt); 
\fill [color=red!80] (8,-0.5) circle (2.5pt); 
\fill [color=red!80] (7,1.7) circle (2.5pt); 
\node [below] at (5.5,0) {$v_0$};  
\node [below] at (8,-0.5) {$v_1$};  
\node [above] at (7,1.7) {$v_2$};  

\node [above] at (6.75,-1.7) {$v_1-v_0+v_2-v_1$};  
\node [above] at (6.75,-2.1) {$+v_0-v_2=0$};  
\end{tikzpicture} %
\end{minipage}}

\protect\caption{$\partial\partial=0$gives us a chain complex.}
\end{figure}
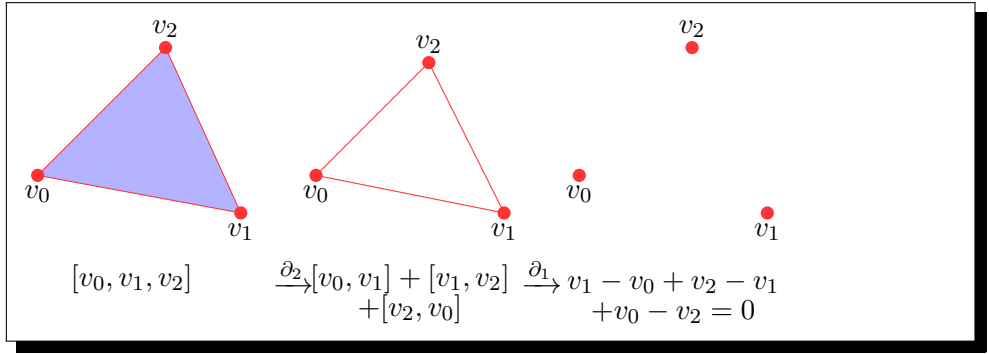

This map gives rise to a chain complex: a sequence of vector spaces
and linear transformations: 
\begin{eqnarray*}
0\stackrel{}{\rightarrow}C_{n}\stackrel{\partial_{n}}{\rightarrow}C_{n-1}\stackrel{\partial_{n-1}}{\longrightarrow}...\stackrel{\partial_{q+2}}{\longrightarrow}C_{q+1}(W)\stackrel{\partial_{q+1}}{\longrightarrow}C_{q}(W)\stackrel{\partial_{q}}{\longrightarrow}...\stackrel{\partial_{2}}{\rightarrow}C_{1}(W)\stackrel{\partial_{1}}{\rightarrow}C_{0}(W)\rightarrow0.
\end{eqnarray*}

It can easily be proved that that for any integer $q$, 
\[
\partial_{q}\circ\partial_{q+1}=0.
\]

In general, a chain complex $C_{\star}=\left\{ C_{q},d\right\} $
is precisely this : a sequence of abelian groups $\left(C_{q}\right)$
connected by an operator $d_{q}:C_{q}\to C_{q-1}$ that satisfies
$d\circ d=0$.

If one defines 
\[
Z_{q}=\ker\partial_{q}\text{ and }B_{q}=\mathrm{im}\,\partial_{q+1},
\]
 then it follows that $B_{q}\subset Z_{q}$. Elements of $Z_{q}=\mathrm{ker}\partial_{q}$
are called cycles, and elements of $B_{q}=\mathrm{im}\partial_{q+1}$
are called boundaries. Likewise, $Z_{q}=\mathrm{ker}\partial_{q}$
is called the $q-$th Cycle Group and $B_{q}=\mathrm{im}\partial_{q+1}$
is called the $q-$th Boundary Group. Then the homology group $H_{q}$
measures the equivalence class of cycles by quotient-ing out the boundaries
i.e. this construction measures how far the sequence is from being
exact.

The $q$-dimensional homology of $W$, denoted $H_{q}(W)$ is the
quotient vector space, 
\begin{eqnarray*}
H_{q}(W)=\frac{Z_{q}(W)}{B_{q}(W)}\cdotp
\end{eqnarray*}
 and the $q$-th Betti number of $W$ is its dimension: 
\[
\beta_{q}=\dim H_{q}=\dim Z_{q}-\dim B_{q}
\]

\section{Morse Homology\label{sec:Morse-Homology}}

Let $F$ be a Discrete Morse function defined on simplicial complex
$W$. Let $C_{q}(W,\mathbb{\,Z})$ denote the space of $q$-simplicial
chains, and $\mathcal{M}_{q}$ which is a subset of $C_{q}(W,\mathbb{\,Z})$
denote the span of the critical $q$-simplices. Let $\mathcal{M_{\star}}$
denote the space of Morse chains. Let $c_{q}$ denote the number of
critical $q$-simplices. Then we have, $\mathcal{M}_{q}\cong\mathbb{Z}^{c_{q}}$.

\begin{theorem}[Forman~\cite{Fo98a}]\label{thm:morsehom}
There exist boundary maps $\widehat{\partial_{q}}\,:\,\mathcal{M}_{q} \rightarrow\mathcal{M}_{q-1}$, for each $q$, which satisfy  \[ \widehat{\partial}_{q}\circ\widehat{\partial}_{q+1}=0. \]  and such that the resulting differential complex  \begin{eqnarray*} 0\stackrel{}{\longrightarrow}\mathcal{M}_{n}\stackrel{\widehat{\partial}_{n}}{\longrightarrow}\mathcal{M}_{n-1}\stackrel{\widehat{\partial}_{n-1}}{\longrightarrow}\dots\stackrel{\widehat{\partial}_{q+2}}{\longrightarrow}\mathcal{M}_{q+1}\stackrel{\widehat{\partial}_{q+1}}{\longrightarrow}\mathcal{M}_{q}\stackrel{\widehat{\partial}_{q}}{\longrightarrow}\dots\stackrel{\widehat{\partial}_{2}}{\longrightarrow}\mathcal{M}_{1}\stackrel{\widehat{\partial}_{1}}{\longrightarrow}\mathcal{M}_{0}\longrightarrow0 \end{eqnarray*} calculates the homology of $W$. i.e. if we go with the natural definition, \begin{eqnarray*} H_{q}(\mathcal{M},\widehat{\partial})=\frac{\mathrm{ker}\widehat{\partial}_{q}}{\mathrm{im}\widehat{\partial}_{q+1}} \end{eqnarray*} \\
Then for each $q$, we have $H_{q}(\mathcal{M},\widehat{\partial})=H_{q}(W,\mathbb{Z})$. 
\end{theorem}

\begin{theorem} [Boundary Operator Computation. Forman~\cite{Fo98a}]\label{thm:boundop} Consider an oriented simplicial complex. Then for any critical (p+1)-simplex $\beta$ set: \\ \\
$\partial \beta = \sum\limits_{critical \: \alpha(p)}\: P_{\alpha\beta}\:\alpha $ \\ \\
$ P_{\alpha \beta} = \sum\limits_{\gamma \in \Gamma(\beta,\alpha)} N(\gamma) $ \\ \\
where $\Gamma(\beta,\alpha)$ is the set of discrete gradient paths which go from a face in  $\eth\,\beta$ to $\alpha$. The multiplicity $N(\gamma)$ of any gradient path $\gamma$ is equal to $\pm1$ depending on whether given $\gamma$ the orientation on $\beta$ induces the chosen orientation on $\alpha$ or the opposite orientation. With the boundary operator above, the complex computes the homology of complex K.  
\end{theorem}

\begin{theorem}[Forman~\cite{Fo98a}] If $a<b$, are real numbers, such that [a,b] contains no critical values of Morse function $\mathcal{f}$, then the sublevel set $\mathcal{M}(b)$ is homotopy equivalent to the sublevel set $\mathcal{M}(a)$.  
\end{theorem}

\begin{theorem}[Forman\cite{Fo98a}]\label{thm:Morsecell} Suppose $\sigma^{p}$ is a critical cell of index p with $f(\sigma)\in[a,b]$ and $f^{-1}(a,b)$ contains no other critical points. Then $M(b)$ is homotopy equivalent to 
\[M(a)\bigcup_{e_b^{p}} e^{p} \]
where $e^{p}$ denotes a p-dimensional cell with boundary $e_b^{p}$.  
\end{theorem}

In Thm.\ref{thm:Morsecell}, Forman's establishes the existence of
a cell complex (let us call it the\emph{ Morse Smale Complex}) that
is homotopy equivalent to the original complex. For proof details
please refer to Forman\cite{Fo98a}. The boundary operator in Thm.\ref{thm:boundop}
for the chain complex construction (referred to as the \emph{Morse
complex}) tells us how to use the new CW complex that is built in
construction described in proof of Thm.\ref{thm:Morsecell}. Note
that the Morse complex itself is a chain complex and not a CW complex.
But, the chain complex construction (referred to as the Morse complex)
tells us that both these constructions have identical homology.

\section{Extra Figures}

\begin{figure}
\shadowbox{\begin{minipage}[t]{1\columnwidth}%
\includegraphics[scale=0.3]{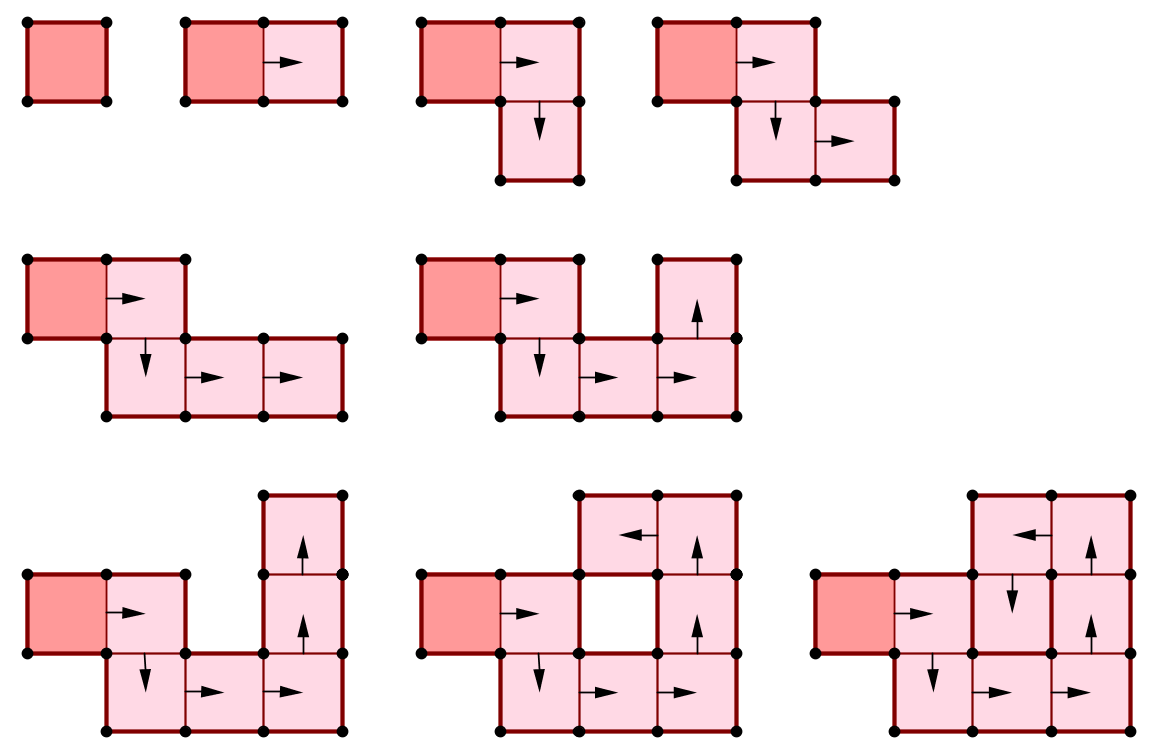}%
\end{minipage}}

\protect\caption{Frame Expansion: Example 1. Part I.}

\label{fig:simple1}
\end{figure}

\begin{figure}
\shadowbox{\begin{minipage}[t]{1\columnwidth}%
\includegraphics[scale=0.35]{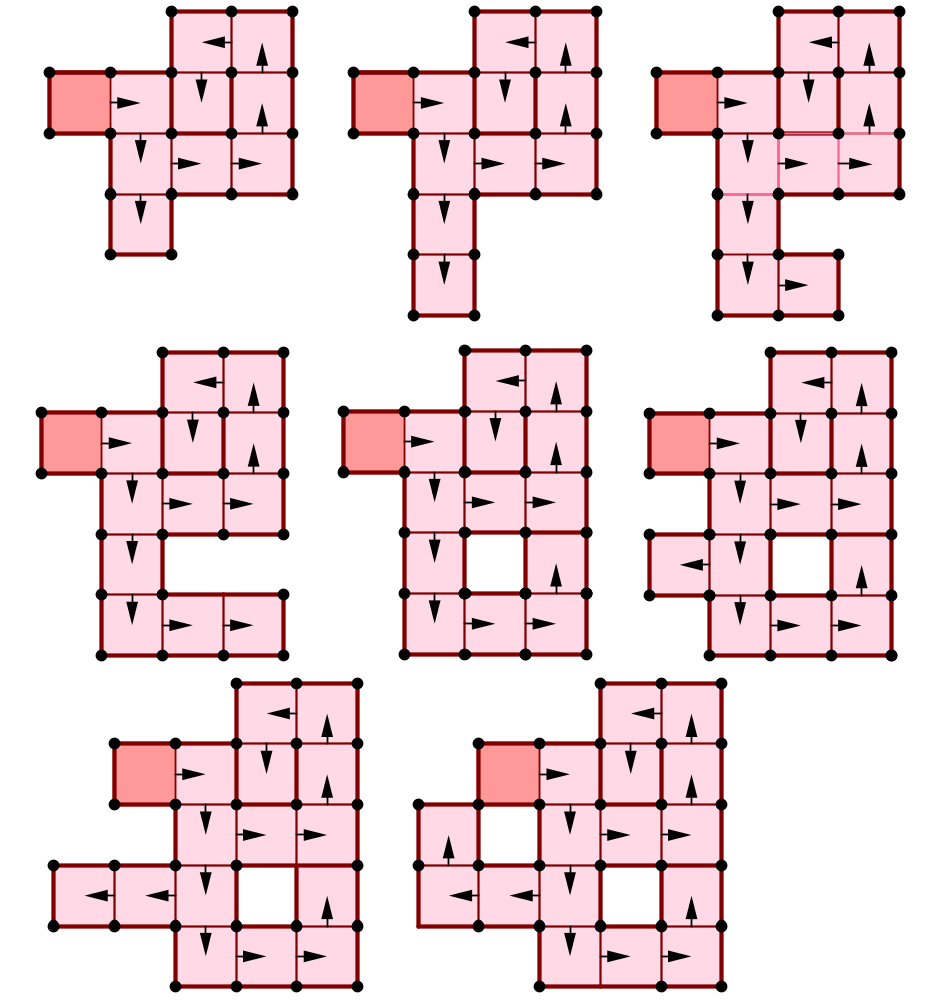}%
\end{minipage}}

\protect\caption{Frame Expansion: Example 1. Part II.}

\label{fig:simple2}
\end{figure}

\begin{figure}

\shadowbox{\begin{minipage}[t]{1\columnwidth}%
\includegraphics[scale=0.25]{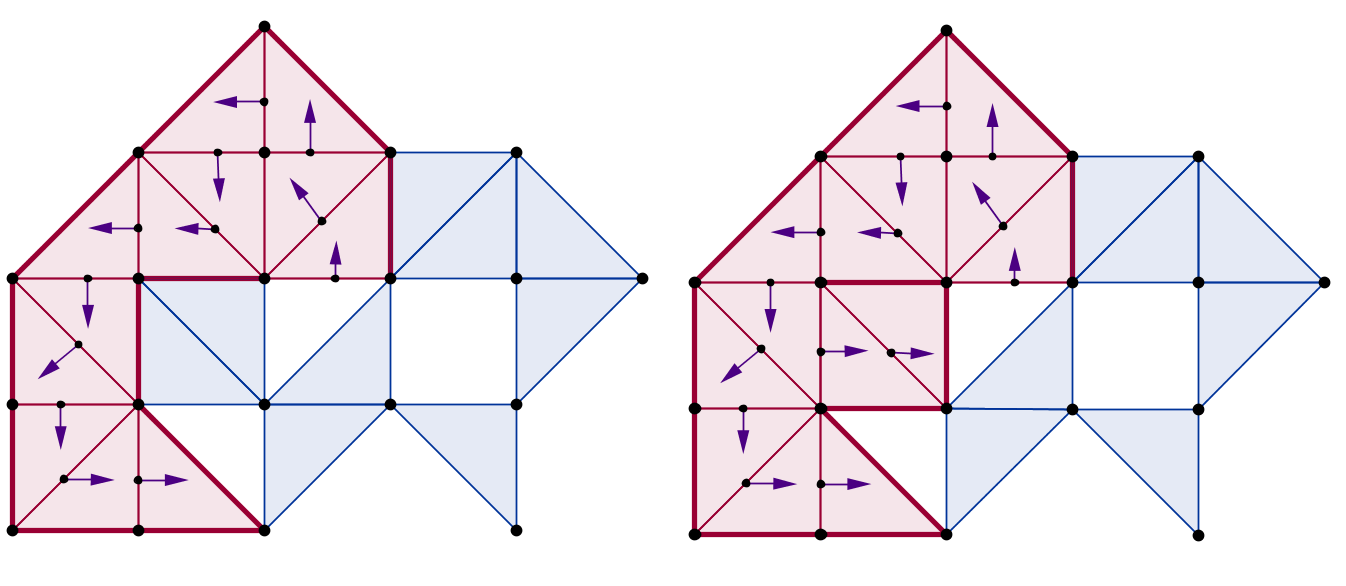}%
\end{minipage}}

\protect\caption{Frame Expansion: Example 2. Part I}

\label{fig:disjoint1}

\end{figure}

\begin{figure}
\shadowbox{\begin{minipage}[t]{1\columnwidth}%
\includegraphics[scale=0.25]{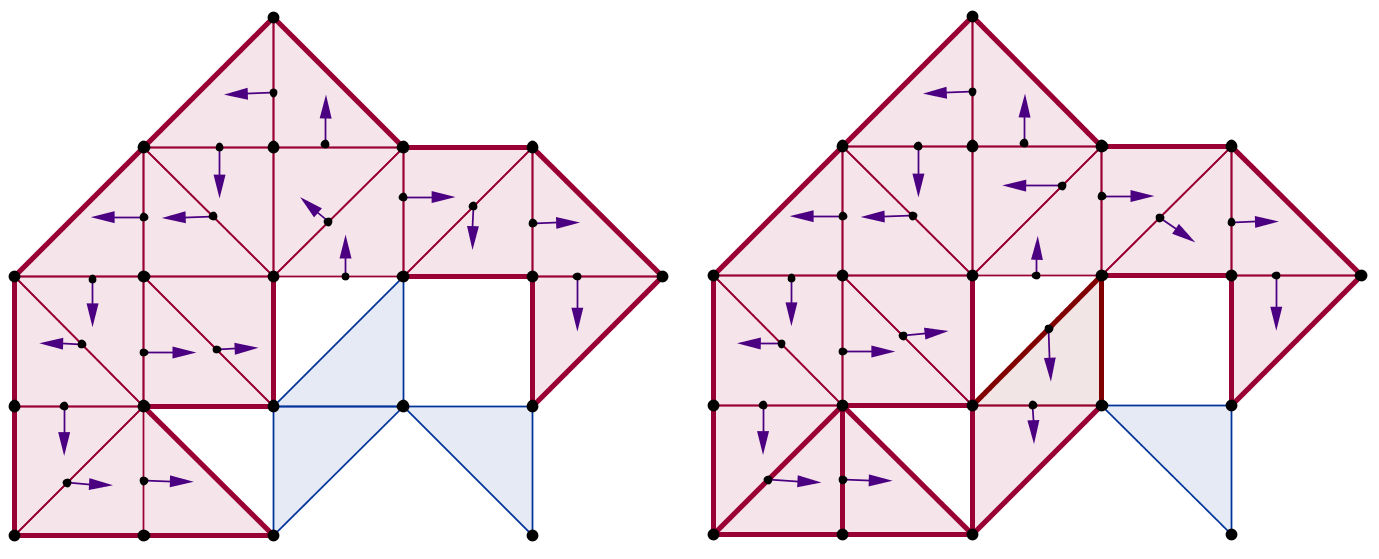}%
\end{minipage}}

\protect\caption{Frame Expansion: Example 2: Part II}

\label{fig:disjoint2}
\end{figure}

\begin{figure}

\shadowbox{\begin{minipage}[t]{1\columnwidth}%
\includegraphics[scale=0.25]{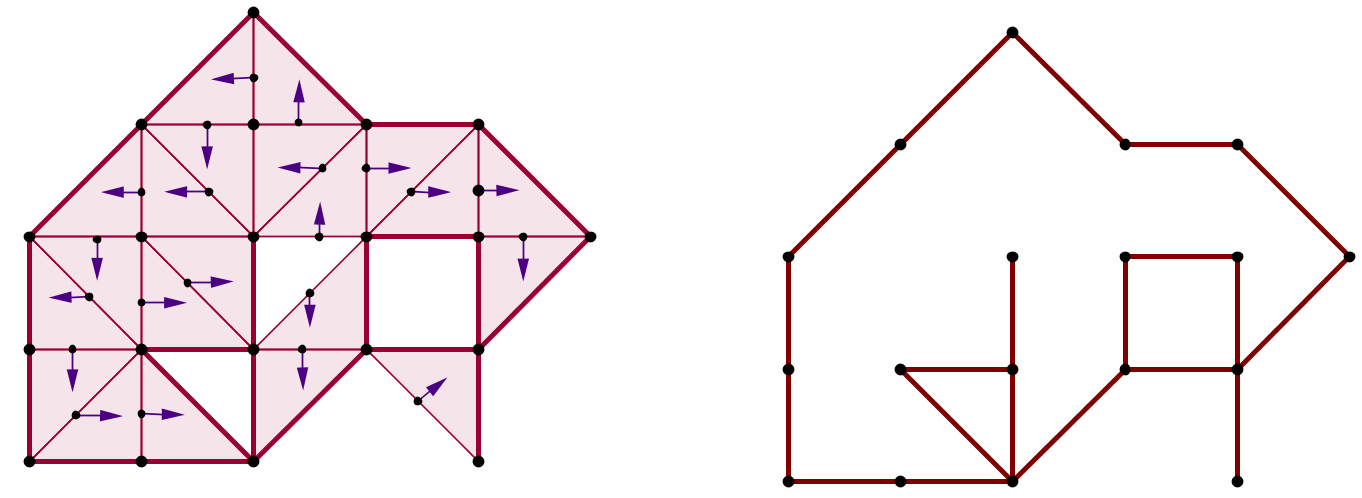}

\label{fig:disjoint3}%
\end{minipage}}

\protect\caption{Frame Expansion: Example 2: Part III}

\end{figure}

\begin{figure}
\includegraphics[scale=0.37]{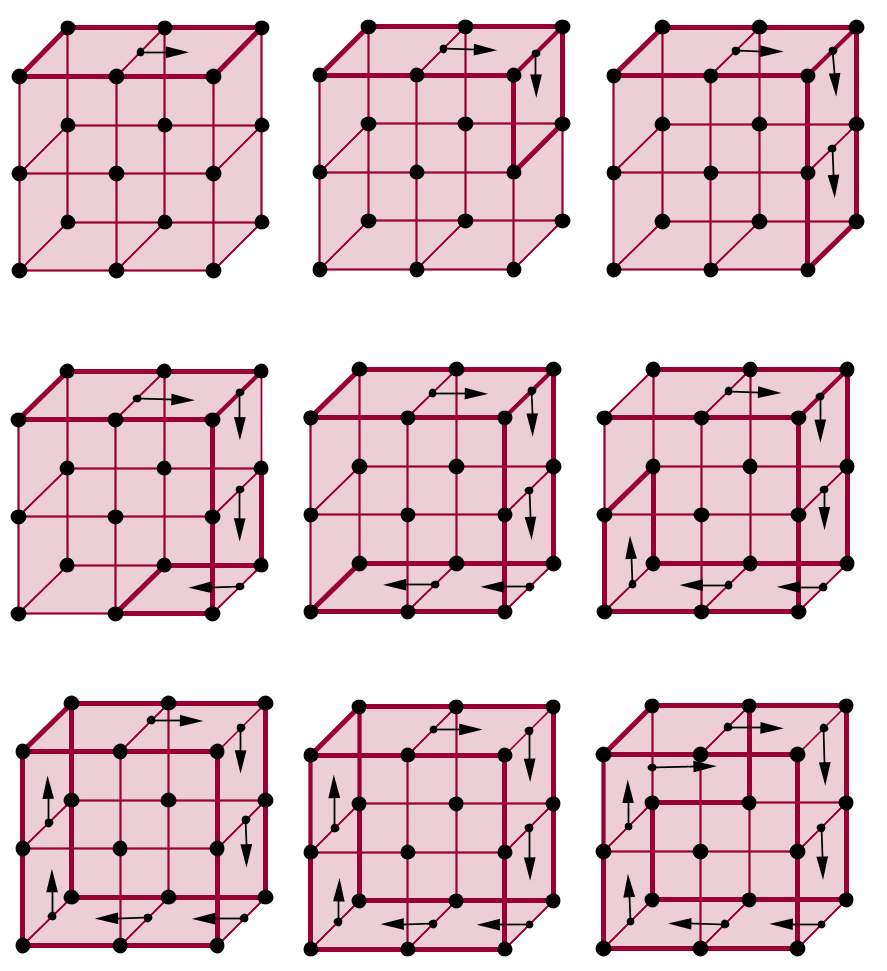}

\protect\caption{Frame Expansion: Example 3: Part 1 }

\label{fig:hollowcube1}
\end{figure}

\begin{figure}
\includegraphics[scale=0.37]{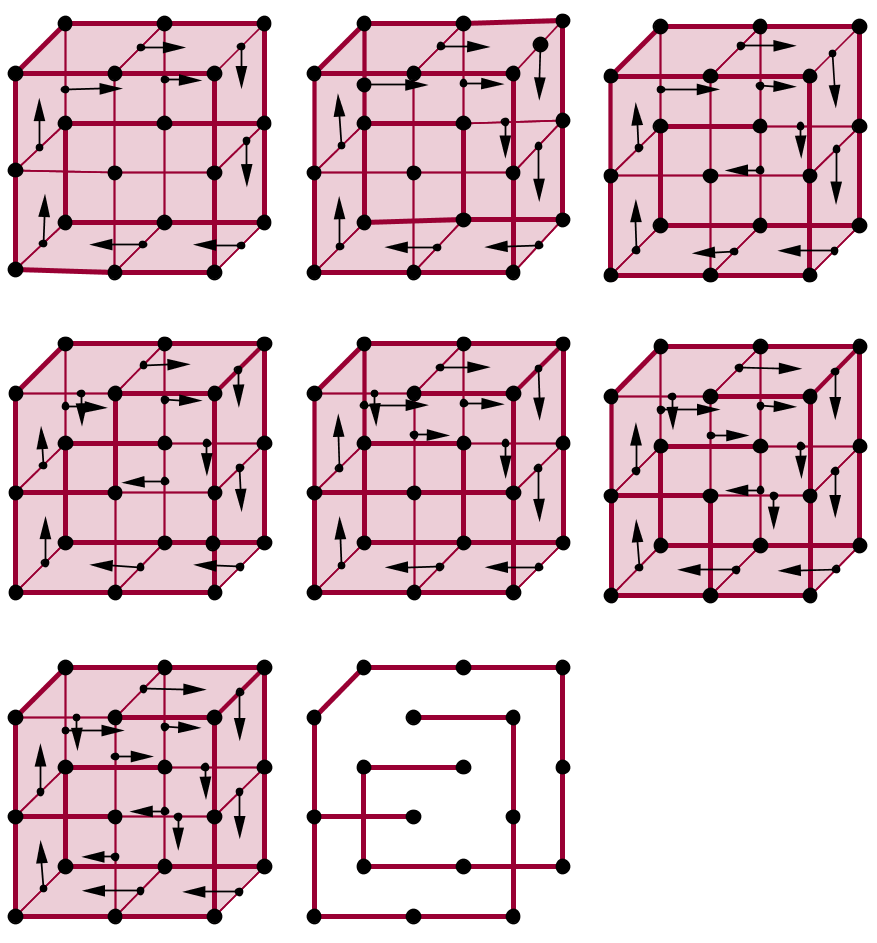}

\protect\caption{Frame Expansion: Exampe 3: Part 2}

\label{fig:hollowcube2}
\end{figure}

\section{Detailed Pseudocode}

\begin{algorithm}
\begin{algorithmic}[1]   

\Procedure{calcHomology}{$\mathcal{M},\mathcal{A}$;}
\LState{$\mathscr{V}\colonequals\tbf{mainFrame}(\mathcal{M})$;}
\LState{$\triangle_c\colonequals\tbf{calcbdryOp}(\mathcal{M},\mathcal{H},\mathscr{V})$;}
\LState{$H(\mathcal{M},\mathcal{A})\colonequals\tbf{SmithNormalForm}(\triangle_c,\mathcal{A})$;}
\EndProcedure

\Statex

\Procedure{calcBdryOp}{$\mathcal{M},\mathcal{H},\mathscr{V}$}
\LState{$\tbf{topologicalSort}(\mathcal{H},\mathscr{V},\mathcal{L},\textsf{'ASCENDING')}$;}
\ForAll{$1 \leq  i \leq \abs{\mathcal{L}}$; $\sigma\colonequals\mathcal{L}[i]$}  	\label{lst:line:loopBS} 	 
	\If{$\dim\sigma=0$ \LAND $\sigma\LargerCdot\tsf{pair}=\tsf{NIL}$}
		\LState{$\triangle \sigma = \emptyset $;}
	\ElsIf{$\sigma\prec\beta$ \LAND $\sigma\LargerCdot\tsf{pair}=\beta$}
		\LState{$\triangle \sigma = <\partial \beta, \sigma> \times \triangle \beta $;}
	\Else
		\LState{Let $\tau_i \prec \sigma$ be the set of regular cells incident on $\sigma$ s.t. $ \langle\tau_i,\sigma \rangle\notin \mathscr{V}$;}
		\LState{Let $\alpha_i \prec \sigma$ be the set of critical cells incident on $\sigma$;}
		\LState{$\triangle \sigma = \sum \triangle \tau_i \times <\partial \sigma, \tau_i>  +  \sum \bm{\alpha_i} \times <\partial \sigma, \alpha_i> $;}
	\EndIf
	\If{$\sigma\LargerCdot\tsf{pair}=\tsf{NIL}$ \LAND $\sigma\LargerCdot\tsf{revPair}=\tsf{NIL}$}
		\LState{$\triangle_c \sigma\colonequals \triangle\sigma$;}
	\EndIf
\EndFor			 		\label{lst:line:loopBE} 	
\EndProcedure 
\end{algorithmic}

\protect\caption{$\textbf{Homology}()$ }
\end{algorithm}

\begin{algorithm}
\begin{algorithmic}[1]   

\Procedure{findCoBdry}{$\mathcal{M}^d,\mathcal{B}^d$}
	\LState{Go through list $\mathcal{M}^{d}$. If $\mathcal{M}^{d}[i]$ has only one face, add $\mathcal{M}^{d}[i]$ to $\mathcal{B}$}
	\LState{$\textbf{return }\mathcal{B}$;}
\EndProcedure  

\Statex

\Procedure{addPairToVectorField}{$\tau, \vartheta,\mathscr{V}, \mathcal{M},\mathcal{B}^d,d$}
\If{$\tau \LargerCdot \tsf{revPair} = \tsf{NIL}$}
	\LState{$\tbf{delete}(\tau,\mathcal{M}^{d}); \tbf{delete}(\vartheta,\mathcal{M}^{d-1})$;}
	\LState \algorithmicif\ {$\tau\in\mathcal{B}$} \algorithmicthen\ {$\tbf{delete}(\tau,\mathcal{B})$;} \algorithmicend\ \algorithmicif
	\LState{$\tbf{nQ} (\mathcal{Q},\tau ) $;\quad$\vartheta\LargerCdot$ pair $\colonequals\tau$;\quad $\tau\LargerCdot$ revPair $\colonequals\vartheta$;} 
	\LState{$\mathscr{V}\colonequals\mathscr{V}+ \langle\vartheta,\tau \rangle$}
\EndIf
\EndProcedure 

\Statex

\Procedure{frameFlow}{$\mathcal{M},d,\mathscr{V}$}

	\If{$(\varpi=\tbf{dQ}(\mathcal{B}^d))$ = NIL}	
		\LState{$\varpi\colonequals\tbf{dQ}(\mathcal{M}^d)$;}
	\EndIf
	\Repeat
	\LState{fF $\colonequals$ 'T';}
		\Repeat
			\LState{$\mathcal{F}\colonequals\tbf{faces}(\varpi)- \varpi\LargerCdot\tsf{revPair}$; $\upsilon\colonequals\textbf{bdry}(\mathcal{F})$;} 
			\If{fF='T' \LAND $\upsilon \neq$ NIL}
				\LState{\textbf{addPairToVectorField}($\varpi,\upsilon,\mathscr{V},\mathcal{M}, \mathcal{B}^d,d$); }
			\Else
				\LState{$\tbf{delete}(\varpi,\mathcal{M}^d)$;}
			\EndIf 
			\LState{fF$\colonequals$'F';}
			\ForAll{$1 \leq  i \leq \abs{\mathcal{F}}$;}  		 			 		
				\LState{$\vartheta\colonequals\mathcal{F}[i]$;  $\mu\colonequals\textbf{cofaces}(\vartheta)-\varpi$; } 
				\LState{$\tbf{addPairToVectorField}(\mu, \vartheta, \mathscr{V},\mathcal{M}, \mathcal{B}^d,d$);}
			\EndFor  
		\Until {$(\varpi= \tbf{dQ} (\mathcal{Q}))\neq \tsf{NIL}$} 
	\Until {$(\varpi=\tbf{dQ}(\mathcal{B}^d))\neq\tsf{NIL}$} 
\EndProcedure

\Statex

\Procedure{mainFrame}{$\mathcal{M},\mathscr{V}$}
	\LState{$\mathcal{B}^2$ \colonequals \tbf{findCoBdry}($\mathcal{M}^2$);}
	\LState{\tbf{frameFlow}($\mathcal{M}$,2,$\mathcal{B}^2$,$\mathscr{V}$);}
	\LState{$\left\{\mathcal{M}^1[1:\tsf{numEars}],\tsf{numEars}\right\}\colonequals\tbf{earDecompose}(\mathcal{M}$);}
	\ForAll{$1 \leq  i \leq \tsf{numEars}$}
		\LState{$\mathcal{B}^1[i]  \colonequals \tbf{findCoBdry}(\mathcal{M}^1[i]$);}
		\LState{\tbf{frameFlow}($\mathcal{M}^1[i]$,1,$\mathcal{B}^1[i]$,$\mathscr{V}$);}
	\EndFor
	\LState{$\tbf{return }\mathscr{V}$.};
\EndProcedure
\end{algorithmic}

\protect\caption{$\textsf{Algorithm }\textbf{FrameFlow}()$ }
\end{algorithm}

\end{document}